%% file: 2026-Rev-Gaps-Budgets-Arxiv.tex
\documentclass[11pt]{scrartcl}

\setkomafont{disposition}{\normalfont\bfseries}

\input{revenue-gaps-macros}

\begin{document}

\title{Evaluating the Performance of Approximation Mechanisms 
under Budget Constraints%
	\thanks{This paper supersedes ``Mind the Revenue 
		Gap: On the Performance of Approximation Mechanisms 
		under Budget Constraints'', by the same 
		authors, that appeared in the \emph{Proceedings of 
		the 17th International Symposium of Algorithmic 
		Game Theory}, \href{https://link.springer.com/chapter/10.1007/978-3-031-71033-9_16}{SAGT 2024},
		Amsterdam, September 3--6, 2024.
		The previous version contained a full proof only 
		for Proposition 2. 
		The current version contains several additional results 
		concerning revenue gaps and revenue non-monotonicity of 
		selling mechanisms under budget constraints, as well as
		full proofs of all the results.  			
		We would like to thank comments and discussions 
		with Murali Agastya, Isa Hafalir, Daniel Lehmann, 
		Simon Loertscher, Idione Meneghel, Noam Nisan and the 
		audience at the 2023 Markets, Contracts and 
		Organizations Conference organized by the Australian 
		National University. 
		Financial support from the Australian Research Council 
		under grant DP190102064 is gratefully acknowledged. } 
	}

\author{
	Juan Carlos Carbajal%
		\footnote{School of Economics, UNSW Sydney, Australia. 
		Email: \href{mailto:jc.carbajal@unsw.edu.au}
		{\nolinkurl{jc.carbajal@unsw.edu.au}}. }  		
	\and 
	Ahuva Mu'alem%
		\footnote{Department of Computer Science, Holon 
		Institute of Technology, Israel.   
		Email: \href{mailto:ahumu@yahoo.com}
		{\nolinkurl{ahumu@yahoo.com}}. }  }

\date{\today}

\maketitle


\begin{abstract} 
\noindent \hrulefill

\noindent{\bfseries Abstract \ } 
We study revenue maximization in a buyer-seller setting where the 
seller has a single object and the buyer has both a private 
valuation and a private budget. 
Private budgets complicate the classic single-product monopoly 
problem, making optimal mechanisms difficult to characterize. 
To address this, we evaluate the robust performance of 
approximation mechanisms relative to optimal mechanisms using 
three performance measures: 
the \emph{guaranteed fraction of optimal revenue}, 
the \emph{maximal value of relaxation}, and a 
\emph{revenue non-monotonicity gap}.
Our analysis reveals sharp contrasts. 
For distributions with bounded support, simple mechanisms with 
polylogarithmic menu size can approximate optimal revenue 
arbitrarily well, even when valuations and budgets are correlated. 
By contrast, for distributions with unbounded support, 
and even for bounded distributions concentrated in the unit square, 
no simple mechanism ---or any mechanism with a finite or 
sublinear menu--- can guarantee a positive fraction of optimal 
revenue. 
In particular, no finite-menu mechanism guarantees any positive 
fraction even under independence. 
We also show unbounded revenue gains from certain relaxations 
under negative correlation and identify cases of revenue 
non-monotonicity. 
Overall, our results show that approximation guarantees with 
private budgets are fragile, revealing fundamental limits to 
simplicity and robustness in mechanism design.



\noindent \hrulefill

\end{abstract}



\section[Intro]{Introduction}
	\label{S:intro}

We study a buyer--seller interaction where the seller has a 
single object to allocate and the buyer has a private valuation 
and a private budget. 
In this paper, we are interested in measuring the performance,
in terms of expected revenue for the seller, of
approximation mechanisms against optimal mechanisms.

The presence of the private budget, which acts as a hard
constraint on the purchasing ability of the buyer, is the
only complication we consider.
Remove it, and we are back in the classic ``single-product
monopoly'' setting with incomplete information.
Previous work in the literature that addresses variations
of this problem has provided strong justifications for the
inclusion of budget constraints.%
	\footnote{See for example \citet{Che:2000ux}, 
    \citet{Che:2013cu}, \citet{Pai:2014cm}, \citet{Richter:2019by}, 
    \citet{Kotowski:2020bh}, and \citet{Boulatov:2021ry}.  }
This inclusion is a promising direction towards more
realism in our theoretical constructs, as
it is hard to think of real-world applications
of auction and market design theory where some sort of
financial restrictions are absent; see
\citet{Salant:1997vx,Cramton:2010txp}.  
Importantly, introducing budget constraints also
points to the limits of our understanding of mechanism
design outside single-item, single-type models.  
For example, \citet{Che:2000ux} consider revenue
maximization by a seller who faces a buyer with private
valuations and private budgets distributed
on a bounded rectangle.
Even when the joint distribution is well-behaved,
the optimal mechanism is extremely complex to obtain and
to understand.
It relies on minor details of the distribution and,
when translated into a menu of lotteries for implementation,
it contains a continuum of such options.

The presence of private budgets is one among the several 
triggers of complexity in mechanism design models.
Indeed, the literature recognized early on that when the 
seller has many items to allocate (the so-called 
``multi-product monopoly'' case), searching for optimal 
mechanisms or even characterizing incentive compatibility 
is challenging.%
	\footnote{The classical references are 
		\citet{McAfee:1988rh} and \citet{Rochet:1998ju}. }
As a result, the seminal work of \citet{Chawla:2007fu}
adopted an alternative route and initiated the study of 
approximation  mechanisms  in the multi-product monopoly case 
without considering budget constraints.
This significant body of research focuses on what is 
often referred to as ``simplicity-optimality tradeoff''.%
	\footnote{\citet{Roughgarden:2019vo} present a recent
		and useful survey. }
The performance of an approximation mechanism is
measured in terms of its expected revenue compared against
the expected revenue of the optimal mechanism.
For instance, \citet{Hart:2017bh, Hart:2019ny} consider 
revenue maximization with $k$ goods and valuations
given by random vectors of the form $(V_1,\ldots,V_k)$.
They show that if $V_1,\ldots,V_k$
are independent random variables, a nontrivial fraction
of the optimal revenue can always be approximated by
simple mechanisms.
Simplicity here refers to the size of the menu associated
with a mechanism.
A mechanism with a smaller number of options or lotteries is
simpler than a mechanism with a larger number of lotteries.
Surprisingly, \citet{Hart:2019ny} find that no such 
``good approximation'' is possible if the components of the 
random vector $(V_1,\ldots,V_k)$ are correlated.
More precisely, their main result is that any mechanism 
with a finite menu size cannot guarantee a positive fraction 
of the optimal revenue for correlated distributions.%
		\footnote{For related results, see \citet{Li:2013er}, 
		\citet{Briest:2015ob}, \citet{Babaioff:2018jy}, 	
		\citet{Rubinstein:2018fc}, \citet{Hart:2019it}, 
		\citet{Babaioff:2020kb}, and \citet{Babaioff:2022by}.}

The dramatic effect that different assumptions have on the
revenue gap between approximation mechanisms and optimal
mechanisms is cause for further study.
A significant contribution of our paper is to
point out how sensitive this revenue gap is to some
modeling assumptions, even in single-item settings, 
as long as the buyer has private budgets
and private valuations.

The preliminaries are presented in \autoref{S:setting}.
We consider a seller with a single item to allocate
(and zero opportunity cost).
The buyer is characterized by a random vector
$(V,W)$, where $V$ refers to the valuation for the
item and $W$ to the budget.
We formalize the seller's problem in terms of direct
mechanisms that satisfy
incentive compatibility, (ex-post) individual rationality,
and (ex-post) budget feasibility.%
    	\footnote{Some work in the literature focus on
    	interim individual rationality and/or interim
    	budget feasibility.
    	We discuss the differences between these concepts in
    	\autoref{S:setting}.}
The set of these ex-post feasible and incentive compatible
mechanisms is denoted by $\cali M$.
The seller's problem is thus to choose a mechanism in $\cali M$
that maximizes expected revenue.
Approximation mechanisms are obtained from $\cali M$
in two different ways:
(i) by imposing additional constraints, thus restricting
$\cali M$ to a subclass of mechanisms, say $\cali N$; or
(ii) by relaxing some original constraints, thus
expanding $\cali M$ to a superclass of mechanisms,
say $\cali N'$.
We measure the performance of a restricted class of
mechanisms $\cali N \subseteq \cali M$, in terms of
expected revenue, using the \emph{guaranteed fraction
of optimal revenue} (GFOR) ratio introduced by 
\citet{Hart:2017bh}.
Intuitively, GFOR measures the `worst-case' revenue loss
incurred by the seller when using mechanisms in the 
subclass $\cali N$ instead of $\cali M$.
We measure the performance of a relaxation class of
mechanisms $\cali N' \supseteq \cali M $ using a related
ratio, which we call the \emph{maximal value of relaxation} 
(MVR).
Intuitively, MVR measures the `best-case' revenue gain
obtained by the seller when using mechanisms in the 
superclass $\cali N'$.
Both ratios perform robust revenue comparisons, i.e.,
over families of $(V,W)$--buyers.
\autoref{S:setting} concludes with a \emph{revenue monotonicity}
result (\autoref{P:rev-mon}) that expands a similar finding 
by \citet{Hart:2015ob} to single-item settings where, in addition 
to the valuation, the buyer has a private budget.
We employ this result to prompt questions on \emph{revenue 
non-monotonicity} in classes of mechanisms that differ 
from $\cali M$, which we also explore.

In \autoref{S:Gaps} we consider different classes of
approximation mechanisms.

\autoref{SS:good-approx} and \autoref{SS:bad-approx} 
focus on \emph{simple} approximation mechanisms.
Following \citet{Hart:2019ny}, we measure the complexity
of a mechanism by the cardinality of its associated menu.
Intuitively, a simpler mechanism is one with a smaller menu.
There are several reasons to focus on simple mechanisms,
notably that the menu size of a mechanism is
related to its communication complexity 
(\citet{Babaioff:2022by}).
In settings without budget constraints, \citet{Myerson1981} 
shows that the optimal mechanism contains one non-trivial 
option.
When the buyer's budget is publicly known, \citet{Chawla:2011mf}
demonstrate that the optimal mechanism has a menu that contains
at most two non-trivial options.
Our main positive result (\autoref{P:good-approx}) establishes 
that, when 
the budget is private, simple mechanisms (i.e., mechanisms 
with polylogarithmic many options) generate an arbitrarily
small revenue gap with respect to optimal mechanisms, 
as long as one considers a family of random vectors $(V,W)$ 
whose support is bounded from above and bounded away from 
zero from below.

In contrast, if the family of $(V,W)$--buyers has unbounded
support or includes random vectors whose support
is a subset of the unit square, then for some family 
of distributions any selling mechanism that contains a
fixed number of options cannot guarantee any positive fraction
of the optimal revenue (\autoref{P:menu-size} and 
\autoref{P:bad-approx-bounded}).
In fact, we strengthen these negative results in two directions. 
First, for some family of finite distributions, any mechanism 
that contains an asymptotically sublinear number of options 
(in the size of the finite support) cannot guarantee a 
positive fraction of the optimal revenue (\autoref{P:sublinear}). 
This suggests a broader way of analyzing approximation gaps, 
one that goes beyond fixed menu-size bounds and instead tracks 
how menu complexity scales with the support of the distribution. 
Second, we obtain an analogue of the ``no good approximation'' 
result of \citet{Hart:2019ny}: in our single-item setting with 
private budgets, there exist infinite-support distributions 
for which no feasible finite-menu mechanism guarantees any 
positive fraction of the optimal revenue, and this holds 
even for certain independent distributions
(\autoref{P:infinite-supp}).

Our positive result on the revenue gaps of simple mechanisms 
presented in \autoref{SS:good-approx} is derived by
combining the revenue monotonicity result mentioned above 
and a discretization of the type space (valuations and 
budgets) that is markedly different from the `nudge-and-round' 
approach introduced by \citet{Hart:2019ny}.
The proofs of our negative results in \autoref{SS:bad-approx} 
are constructive.
To conclude this part of the paper, we strengthen the 
impossibility result of \citet{Che:2000ux} 
by showing that, for a single item with private values, 
the requirement of infinite menus is not limited to 
optimal mechanism. 
Rather surprisingly, this requirement extends to any 
approximately optimal mechanism as well.

We switch gears in \autoref{SS:cash-bond} and pay attention
to a relaxation of the seller's problem, where the seller
can costlessly prevent the buyer from over-reporting
her budget.%
		\footnote{There are several ways the
		seller can prevent over-reporting.
		For instance, the buyer could be required to
		post a cash-bond prior to the interaction
		with the seller. }
Our analysis here is motivated by both computational and
economic considerations, which we discuss in detail later on.%
		\footnote{Previous work considering this relaxation
		include \citet{Devanur:2017ik} and 
		\citet{Daskalakis:2018zg}. }
\citet{Che:2000ux} show that if valuations 
and budgets are positively correlated, then this relaxation to
mechanisms where the seller ignores deviations to
budget over-reporting has no revenue advantage for the seller.
In contrast, we show that when valuations and budgets
are negatively correlated, the value of relaxation can be
unboundedly large (\autoref{P:negative-aff}).  
In addition, we show that a cash-bond relaxation on the 
set of ex-post feasible and incentive compatible mechanisms 
$\cali M$ can be subject to \emph{revenue non-monotonicities}: 
a $(V,W)$--buyer may first-order stochastically dominate 
a $(V',W')$--buyer, yet the expected revenue collected by 
the seller from the latter may be larger than from the former. 
In fact, we show that revenue non-monotonicities in this 
setting can be acute (\autoref{P:rev-non-mon-CB}).

Finally, in \autoref{SS:sic} we consider a restricted
problem in which the seller considers deviations from
truthful reporting even when these deviations may lead
to unaffordable choices for the buyer.
The imposition of these stronger incentive constraints
has received attention in the literature because it
provides computationally feasible solutions  
(\citet{Daskalakis:2018zg}).
We show that the optimal restricted mechanisms perform
badly in comparison with the optimal mechanisms
when the family of $(V,W)$--buyers has unbounded support
(\autoref{P:menu-SIC}) or has support restricted
to the unit square (\autoref{P:menu-SIC-square}).
As in the previous setting we consider, here we also 
obtain a revenue non-monotonicity result 
(\autoref{P:rev-non-mon-SIC}).

Overall, our research highlights the limitations of using
approximation mechanisms in settings with private valuations
and private budgets.
Approximation mechanisms yield a negligible revenue gap 
in only one of the six cases we consider.   
Specifically, we show that for distributions with bounded
support (above and below away from zero), an arbitrarily
high fraction of the optimal revenue can be obtained
using a simple mechanism with polylogarithmic menu size. 
This result applies even if the valuation and the budget 
are arbitrarily correlated. 
In all others, the revenue gap between the optimal mechanism
and the best performing approximation mechanism is 
unboundedly large.  
Thus, despite some clear advantages (including computational)
in dealing with 
restricted or relaxed classes of mechanisms, one has to be
aware of large potential revenue losses.
Our research also stresses the crucial role some of our
assumptions play in the performance of our theoretical
constructs.
Our results, positive and negative, apply to a single-item 
setting.  
An interesting direction for future work is to investigate 
whether our main positive results can be computed in 
polynomial time, not only for the single-item case but 
also in richer settings with multiple bidders and 
multiple items.

\subsection*{Related Work}

For single-item auctions, \citet{Laffont:1996ig} show that when 
there is a single buyer with a publicly known budget and a 
value drawn from a regular distribution or DMR distributions, 
the optimal mechanism has a menu size of one: it posts a 
price equal to the minimum of the buyer's budget and the 
Myersonian reserve price.  
When the budget is publicly-known and the value of the single 
bidder is drawn from a general distribution, \citet{Chawla:2007fu} 
show that the optimal mechanism has a menu of size two. 
For multiple bidders with public budgets, 
\citet[Theorem 7]{Chawla:2007fu} give a 2-approximation mechanism. 
\citet[Theorem 16]{Chawla:2007fu} also prove that, for 
private-budget MHR distributions (where valuations and budgets 
are independent) with multiple bidders, there exists a 
$3(1+e)$-approximation mechanism.
\citet{Devanur:2017ik} characterize the optimal mechanisms in 
a relaxed setting where a single buyer is unable to overstate 
her budget. 
They show that the optimal mechanism may require a menu with 
exponentially many options (in the number of possible budgets).

For the multi-item revenue maximization problem with a single 
budget-constrained additive buyer, \citet{Cheng:2021ew} 
analyze two simple pricing schemes: selling all items as 
a single bundle, and selling each item separately. 
They show that the better of these two mechanisms guarantees 
a constant fraction of the optimal revenue when the budget 
is drawn from an independent distribution satisfying the 
monotone hazard rate property.%
	\footnote{The single-item case can be viewed as a special 
		case of their framework (with one good). 
		In this degenerate setting, both bundling and separate 
		sales coincide with a single posted price, so their 
		constant-factor guarantee implies that menu size 
		one suffices to approximate the optimal revenue in 
		single-item auctions satisfying the monotone hazard
		rate property with private budgets and private values 
		drawn independently.}  
This builds on the earlier influential paper by 
\citet{Babaioff:2020kb}, who study the multi-item setting 
without budget constraints and show that the same simple 
pricing approach yields a constant-factor approximation to 
the optimal revenue. 
\citet{Rubinstein:2018fc}  extend the result of 
\citet{Babaioff:2020kb} to the multi-item, 
subadditive buyer revenue maximization problem (still without 
budgets). 
Interestingly, although revenue monotonicity does not hold 
in general for this setting (as shown by \citet{Hart:2015ob}),
\citet{Rubinstein:2018fc} show that improved bounds on 
approximate revenue monotonicity directly lead to stronger 
revenue guarantees for simple mechanisms.

In contrast, our results focus on single-item approximation 
mechanisms with private values and private budgets, allowing 
arbitrary correlation between value and budget.  
On the positive side (\autoref{P:good-approx}), we establish 
that for a single-item auction, an arbitrarily high fraction 
of the optimal revenue can be achieved using a polylogarithmic 
menu size.
This holds for the general case in which valuation and budget 
can be arbitrarily correlated, with the only assumption being 
that the support of the distribution is bounded above and
away from zero. 
On the negative side (\autoref{P:menu-size} to 
\autoref{P:infinite-supp}), we prove fundamental
impossibility results for distributions with unbounded support. 
In this setting, we show that achieving any positive fraction 
of the optimal revenue is impossible for finite distributions 
with a sublinear menu, and for infinite distributions with 
any finite menu. 
The sublinear-menu result is, to the best of our knowledge, 
new and reflects a broader way of analyzing approximation 
gaps, one that applies not only to fixed menu-size classes 
but also to classes whose menu complexity grows with the 
support size.
Some of our negative results apply even to independent 
distributions (e.g., \autoref{P:bad-approx-bounded} and 
\autoref{P:infinite-supp}). 
This stands in sharp contrast to the multi-item setting 
(with independent, additive values) without budgets, where
\citet{Babaioff:2022by} show that an arbitrarily high fraction 
of the optimal revenue can be extracted with a finite menu 
(even if the support is unbounded).

Our negative results are also related to the literature on 
selling multiple correlated goods. 
\citet{Briest:2015ob} show that in a unit-demand environment 
with correlated values, lottery pricing can outperform 
deterministic  pricing by an unbounded factor. 
\citet{Hart:2019ny} later show that, for correlated valuations, 
a similarly dramatic failure of simple mechanisms already 
occurs  with just two goods, thereby resolving an open problem 
left by \citet{Briest:2015ob}.  
Building on this, \citet{Psomas:2022rd} recently showed that 
\citet{Hart:2019ny} sufficient condition is in fact necessary 
for establishing infinite separations between simple and 
optimal mechanisms for multiple correlated goods. 
Although our environment differs, \autoref{P:infinite-supp}
6 establishes an analogous impossibility in our single-item 
model with private budgets: even here, no finite-menu mechanism 
guarantees any positive fraction of the optimal revenue; 
remarkably, this remains true even for certain independent 
distributions.


\section[Setting]{Setting}  \label{S:setting}

We study a buyer--seller interaction where the seller has
one object to sell and the buyer has a private valuation
and a private budget.
This last element is the only complication we consider.
Without the presence of a hard financial constraint on the
buyer's side, we are back in the canonical single-product
monopoly setting in which the revenue maximization problem
is thoroughly understood.
Our focus on the single-product case allows us to highlight
the limits of using approximation mechanisms in the
presence of budget constraints.

The seller's objective is to maximize expected revenue
---the seller's opportunity cost is zero.
The buyer has a private valuation $v \geq 0$ for the object and
a private budget $w \geq 0$ constraining her ability to pay.
Valuations and budgets are given by a random vector
$(V,W)$ that takes values in $\R_+^2$.
We do not exclude the possibility of $V$ and $W$ being
correlated.
The realization of this random vector is private information
of the buyer.
The seller only knows the distribution of the $(V,W)$--buyer
involved in the exchange.%
	\footnote{We use the expressions \emph{$(V,W)$--buyer}
	and \emph{random vector $(V,W)$} interchangeably.  }

The seller offers \emph{lotteries} of the form 
$(q,P_{e})$. 
Here $q \in [0,1]$ denotes the probability that the buyer 
gets the object, and $P_{e} \geq 0$ denotes the price the buyer 
pays in case the good exchanges hands  ---if there is no 
exchange, the buyer pays nothing.  
The expected payment to the seller generated by lottery 
$(q,P_e)$ is thus $P_e\,q$. 
A \emph{menu} is a collection (finite or infinite) 
of lotteries offered by the seller.


\subsection[Buyer]{Buyer's Behavior}  \label{SS:buyer}

We consider a buyer who, due to her hard budget constraint,
employs a two-stage approach in deciding which lottery
to choose from those available in a given menu $M$.
In the first stage, the buyer makes a shortlist composed
of all lotteries in $M$ that
are \emph{ex-post individually rational} and \emph{ex-post
budget feasible} given her budget.
In the second stage the buyer selects, from this shortlist,
a lottery that maximizes her expected utility.

Formally, given menu $M$, a $(V,W)$--buyer with type 
realization $(v,w) \in \R_+^2$ sorts out the collection of 
lotteries  
	\begin{equation*}
		M(v,w) 
		\ = \ \big\{ (q,P_{e}) \in M \, \colon \,  
			P_{e} \, \leq \, \min (v,w)  \big\} .
	\end{equation*}
Any lottery in $M(v,w)$ is ex-post budget feasible for the 
buyer ---she is always able to afford the price associated 
to either outcome, trade or no trade--- and ex-post 
individually rational ---she weakly prefers purchasing any 
lottery in $M(v,w)$ to not interacting with the seller, 
trade or no trade. 
Note that the value of the buyer's outside 
option is normalized to zero.  
From the shortlist $M(v,w)$, the buyer selects a lottery 
$(q^*,P^*_{e})$ that maximizes her expected 
utility; i.e., 
	\begin{equation*}
		(q^*,P^*_e) \ \in \ 
		\arg\max \, \big\{ (v - P_{e})\,q 
        \, \colon \, (q,P_{e}) \, \in \, M(v,w) \big\}.
	\end{equation*}  
It is convenient to let $p = P_{e} \, q$ denote the 
\emph{expected payment} to the seller.  
We now write lotteries simply as pairs $(q,p)$, where 
$q \in [0,1]$ and $p \geq 0$. 
The \emph{actual payment} from the buyer to the seller, 
conditional on trade, is of course $p/q$ as long as 
$q > 0$, and zero otherwise.

Notice that the buyer's behavior is presented as a two-stage 
decision process to highlight the fact that the buyer is an 
expected utility maximizer and has quasi-linear preferences.  
This is reminiscent of certain sequential elimination 
processes employed in behavioral economic theory to study 
choices by a boundedly rational agent 
(see \citet{Manzini:2007oy}). 
In our case the buyer is not boundedly rational, merely 
budget constrained.


\subsection[Seller]{Seller's Problem}  \label{SS:seller}

By the Revelation Principle (\citet{Myerson1981}), 
we formalize the seller's problem in terms of direct 
mechanisms.
A \emph{direct (selling) mechanism} $\mu = (x,s)$ is
composed of a pair of Borel measurable functions,
where $x \colon \R_+^2 \to [0,1]$ describes the probability
of the good being allocated to the buyer, and 
$s \colon \R_+^2 \to \R_+$ describes the expected payment
to the seller.
To be more explicit, given a report $(\tilde v, \tilde w)
\in \R_+^2$ from the buyer, the mechanism $\mu = (x,s)$
assigns the object with probability $x(\tilde v, \tilde w)
\in [0,1]$ in exchange for an expected payment of
$s(\tilde v, \tilde w) \geq 0$ to the seller.

Given the buyer's behavior, the seller takes into account
the following restrictions in the design of direct mechanisms.
First, $\mu = (x,s)$ must satisfy the \emph{ex-post budget
feasibility} constraint; i.e., for all $(v,w) \in \R_+^2$,
	\begin{equation}
		s(v,w)  \ \leq \ w \,x(v,w) . 	
			\tag{BF}  \label{Eq:BF}
	\end{equation}
(Recall the buyer's actual payment upon receiving
the object is $s(v,w)/x(v,w)$ as long as $x(v,w) > 0$,
and zero otherwise.)
Second, the mechanism $\mu = (x,s)$ must be
\emph{ex-post individually rational}; i.e., for all
$(v,w) \in \R_+^2$,
	\begin{equation}
		v \, x(v,w) \, - \, s(v,w) \ \geq \ 0 .  		
			\tag{IR} 	\label{Eq:IR}   	
	\end{equation}
(Recall that the buyer pays zero if the object doesn't
exchange hands.)
Finally, the mechanism $\mu = (x,s)$ must also be \emph{incentive
compatible}: for all $(v,w) \in \R_+^2$ and all
$(\tilde v, \tilde w) \in \R_+^2$ such that $s(\tilde v,\tilde w)
\, \leq \, w \, x(\tilde v, \tilde w)$, it must be that
	\begin{equation}
		v \, x(v,w) \, - \, s(v,w)
		 \ \geq \ v \, x(\tilde v,\tilde w) \, - \,
			s (\tilde v,\tilde w) .  \tag{IC} \label{Eq:IC}
	\end{equation}

In words, the buyer participates in the mechanism
and truthfully reveals her type if (i) the ex-post
payment upon receiving the object is less than her budget,
and the ex-post payment when trade does not occur is zero;
(ii) the ex-post utility and, thus the expected utility,
associated to participating in the mechanism is weakly greater
than the value of her outside option, which is normalized
to zero; and
(iii) the expected utility from truthfully revealing her
type to the mechanism is weakly greater than
the expected utility associated with any other report,
as long as this deviation is ex-post affordable for the buyer,
given her true type.%
		\footnote{Crucially, because of this budget affordability
		requirement, the incentive constraints in this
		case do not allow the seller's problem to
		be expressed as a linear program.}
We let $\mathcal{M}$ denote the class of all selling mechanisms
defined on $\R_+^2$ that satisfy the \eqref{Eq:BF},
\eqref{Eq:IR} and \eqref{Eq:IC} constraints.

Our focus on ex-post budget feasibility and ex-post
individual rationality is not uncommon in the literature
on auctions and mechanisms with budget constrained agents.
The analysis of first-price and second-price auctions under
budget constraints is usually performed
under these two assumptions ---see for instance 
\citet{Che:1998je}, \citet{Chawla:2011mf}, and  
\citet{Kotowski:2020bh}.
Some work on efficient and revenue-maximizing resource
allocation mechanisms with many agents also focuses on
ex-post constraints, e.g., \citet{Che:2013cu} and 
\citet{Richter:2019by}.
On the other hand, \citet{Che:2000ux} and \citet{Pai:2014cm}
focus on interim constraints in their analyses.
The ex-post constraints make more sense in our setting, where
we focus on mechanisms that offer lotteries to a
single buyer.  
\citet{Che:2013cu} argue that, while intellectually interesting,
lotteries with positive entry fees are not common in the
real world because they are susceptible to manipulation
by the seller.
Focusing on ex-post instead of interim constraints is
not without consequences, as expected revenue generally
will be lower under the ex-post constraints.%
		\footnote{\autoref{Ex:public-v} provides a concrete 
		illustration ---see also \citet{Chawla:2011mf}. }

Observe we define a direct mechanism $\mu = (x,s)$ on $\R_+^2$ 
even when the random vector $(V,W)$ may not have full support.  
This is without loss of generality: our \autoref{L:extension} in
\hyperlink{App-A}{Appendix A} shows that 
any ex-post budget feasible and incentive compatible 
mechanism defined on the support of $(V,W)$ can be extended 
to an ex-post budget feasible and incentive compatible 
mechanism defined on $\R_+^2$. 
From there, it is immediate to show that 
if the original mechanism defined on the support of 
$(V,W)$ is also ex-post individual rational, then so is its 
extension to $\R_+^2$. 
\citet{Hart:2015ob} show a similar \emph{Extension Lemma} in 
a setting without budget constraints, but with multiple goods. 
To the best of our knowledge, we are the first to point out its 
validity under hard budget constraints.

The expected revenue raised by $\mu = (x,s)$
when the seller interacts with a $(V,W)$--buyer is 
$R(\mu;V,W) := \mathbb{E}[s(V,W) ]$.  
Facing this buyer, the \emph{seller's problem} is to find a 
direct mechanism $\mu \in \mathcal{M}$ that maximizes 
expected revenue.
The \emph{optimal revenue from the $(V,W)$--buyer} is
the value of the solution to the seller's problem
at $(V,W)$, if one exists; i.e.,
    \begin{equation}
        \Rev(V,W) 
        \ := \ \sup_{\mu \in \mathcal{M}} R(\mu;V, W) .
    \end{equation}
An \emph{optimal mechanism for the $(V,W)$--buyer} is any
mechanism $\mu^*$ in $\mathcal{M}$ that generates expected
revenue $\Rev(V,W)$.


\subsection[Estimates]{Estimating the Revenue Gaps}
	\label{SS:estimates}

Several papers in the literature have studied variations
of the seller's problem.
This body of work points to the complexity of optimal mechanisms 
under budget constraints, or more generally under 
multi-dimensional type spaces. 
As a result, important contributions in the literature focus on
solving the revenue maximization problem over a more
restricted set of (simpler) mechanisms or a more
relaxed set of mechanisms.
In this paper, we refer to feasible solutions of these
related problems as \emph{approximation mechanisms} and
explore the limits of using approximation mechanisms to
solve the seller's revenue maximization problem.
When focusing on a different class of mechanisms $\cali N \ne
\cali M$, we write
\begin{equation*}
		\Rev (V,W |\cali N)
		\ := \ \sup_{\nu \in \cali N} R(\nu;V,W)
\end{equation*}
for the optimal revenue that can be raised by a seller
who interacts with a $(V,W)$--buyer and uses mechanisms
in the class $\cali N$.
An \emph{optimal approximation mechanism in the
class} $\cali N$ is any mechanism $\nu^* \in \cali N$
that generates expected revenue $\Rev(V,W | \cali N)$.

A branch of the literature considers restricted classes
of mechanisms that are `simpler' than $\cali M$ and thus
have lower communication complexity or computational complexity.
This work focuses on the expected revenue that is lost
by using simple mechanisms.
To measure any potential losses, we follow \citet{Hart:2017bh}
and focus on worst-case scenarios.
In what follows, let $\scr B$ denote a given family of
$(V,W)$--buyers that take values on $\R_+^2$.  
To avoid any measure-theoretic complications, we assume that 
all random vectors in $\scr B$ are defined on the same 
probability space.  
When specific properties of buyers in $\scr B$ are required 
(e.g., bounded support, correlation between valuations and 
budgets), we will be explicit about it.

\begin{definition}  \label{D:GFOR}
Let $\scr B$ be a given family of $(V,W)$--buyers and
$\cali N \subseteq \cali M$ be a non-empty subclass of
mechanisms.
The \emph{Guaranteed Fraction of Optimal Revenue} for
$\cali N$, denoted by $\GFOR (\cali N; \scr B)$, is
defined as the ratio
	\begin{equation*}
		\GFOR (\cali N; \scr B)
			\ := \ \inf_{(V,W) \in \scr B} \frac{\Rev ( V,W
				| \cali N)}{\Rev (V,W)} .
	\end{equation*}
\end{definition}

\smallskip

In words, $\GFOR (\cali N; \scr B)$ is the maximal
number $0 \leq \alpha \leq 1$ such that, for every
$(V,W)$--buyer in $\scr B$, there is a mechanism $\nu$ in
the class $\cali N$ that generates at least $\alpha$
times the optimal revenue $\Rev (V,W)$.  
\citet{Hart:2017bh} discuss the importance 
of this concept in the multiple-goods monopoly model 
and its relationship with the \emph{competitive ratio} 
concept in the computer science literature.
When $\alpha$ is close to $1$, the revenue gap between
an optimal approximation mechanism in the subclass $\cali N$
and an optimal mechanism in $\cali M$ is negligible,
and thus $\cali N$ can serve as a desirable replacement
for $\cali M$.

A related branch of the literature considers relaxing
some of the constraints embedded in the class $\cali M$
of \eqref{Eq:IC}, \eqref{Eq:IR} and \eqref{Eq:BF}
mechanisms.
Sometimes this relaxation is motivated by computational
requirements ---e.g., weakening some of the constraints
transforms the seller's problem into a computationally
feasible linear program.
In other instances, a relaxation is motivated by economic
or institutional considerations.  
\citet{Che:2000ux} explore revenue maximization
when the seller can force the buyer to set up a cash bond,
thus preventing the buyer from over-reporting her budget.
We consider best-case scenarios in measuring any potential
gains from considering a relaxed environment.

\begin{definition}  \label{D:MVR}
Let $\scr B$ be a given family of $(V,W)$--buyers and
$\cali N' \supseteq \cali M$ be a non-empty superclass of
mechanisms.
The \emph{Maximal Value of Relaxation} for $\cali N'$,
denoted by $\MVR (\cali N'; \scr B)$, is the ratio
	\begin{equation*}
		\MVR (\cali N'; \scr B)
		\ := \ \sup_{(V,W) \in \scr B} \frac{\Rev ( V,W
			| \cali N')}{\Rev (V,W)} .
	\end{equation*}
\end{definition}

Thus, $\MVR (\cali N'; \scr B)$ is the minimal number
$\beta \geq 1$ such that, for every $(V,W)$--buyer in the
family of random vectors $\scr B$, every mechanism $\nu'$ in
the class $\cali N'$ generates at most a multiple $\beta$ of
the optimal revenue $\Rev (V,W)$.
When $\beta$ is close to $1$, the revenue gap between
an optimal mechanism in $\cali M$ and an optimal approximation
mechanism in the superclass $\cali N'$ is negligible,
and so is the value of relaxing the seller's problem.

An additional reason to focus on $\GFOR$ and $\MVR$ is that they
provide a measure of robustness under higher orders of
uncertainty.
The seller knows that a buyer is characterized by the
random vector $(V,W)$ in $\scr B$, and thus is aware of the
fact that facing a
$(V,W)$--buyer is different from facing a $(V',W')$--buyer.
However, the seller may not know which buyer is present
at the time of exchange and thus may be interested in
prior-free measurements of the revenue gap.
When considering the use of simpler mechanisms in $\cali N$,
the seller focuses on a `worst-case' measure of lost
revenue, which is precisely what $\GFOR (\cali N; \scr B)$
provides.
When considering the value of relaxing some of the constraints
in the revenue maximization problem, the seller focuses on
a `best-case' measure of gained revenue, which is what
$\MVR (\cali N'; \scr B)$ represents.


\subsection[Rev Mon]{Revenue Monotonicity} \label{SS:rev-mon}

Before we present our contribution regarding the performance 
of approximation mechanisms, we show an important revenue 
monotonicity property satisfied by mechanisms in $\cali M$: 
under a reasonable tie-breaking rule, expected revenue increases
whenever the valuation and/or the budget of the buyer increases.
This holds for any mechanism satisfying \eqref{Eq:IC},
\eqref{Eq:BF} and \eqref{Eq:IR}, not just revenue-maximizing
mechanisms.
Following \citet{Hart:2015ob}, define \emph{seller-favorable
mechanisms} as those that break ties in favor of the seller.
Since we are interested in the revenue gaps between
optimal mechanisms and optimal approximation mechanisms,
focusing on seller-favorable mechanisms does not entail any
loss of generality.

Recall that a random vector $(V',W')$ first-order 
stochastically dominates $(V,W)$ if and only if 
$\mathbb{E} \, [u(V',W')] \, \geq \, \mathbb{E} \, [u(V,W)]$
for any non-decreasing function $u \colon \R^{2} 
\to \R$.  
If $(V,W)$ and $(V',W')$ belong to a family of buyers 
$\scr B$ (and thus are defined on the same probability space), 
then $(V',W')$ first-order stochastically dominates 
$(V,W)$ if and only if $(V',W') \geq (V,W)$ 
---see \citet{Shaked2007rt}.

\begin{proposition}  \label{P:rev-mon}
Suppose that $\mu = (x,s) \in \cali M$ is a seller-favorable
mechanism.
For any two buyers $(V,W)$ and $(V',W')$ such that $(V',W')$
first-order stochastically dominates $(V,W)$, one has
	\begin{equation*}
		R(\mu;V,W) \ \leq \ R(\mu;V',W') .
	\end{equation*}
\end{proposition}

\begin{proof}
Fix a seller-favorable mechanism $\mu = (x,s) \in \cali M$.
It suffices to show that the payment function $s$
is weakly increasing, in the sense that
for any two buyer's types $(v,w)$ and $(v',w')$ such
that $(v,w) \leq (v',w')$, one has $s(v,w) \leq s(v',w')$.
This is so because, if that were the case, when $(V,W)$ is
first-order stochastically dominated by $(V',W')$ we have
	\begin{equation*}
		R(\mu;V,W) \ = \ \mathbb{E}\big[s(V,W)\big]
		\ \leq \ \mathbb{E}\big[s(V',W')\big] \ = \ R(\mu;V',W') .
	\end{equation*}

We divide our argument in two steps.
The first one is similar to traditional arguments made
in settings without budget constraints that show that
allocations and payments are weakly increasing in the
buyer's valuation.%
	\footnote{See for instance \citet{Hart:2015ob}.
		We include the proof of this step for completeness. }
The second step is novel and specific to our budget-constrained 
setting.
Consider two distinct type realizations
$(v,w)$ and $(v',w')$, such that $(v,w) \leq (v',w')$.
Let $(q,p) = (x(v,w),s(v,w))$ and $(q',p') = (x(v',w'),
s(v',w'))$ denote the lotteries selected by these types
under the selling mechanism $\mu = (x,s)$ in $\cali M$.

\smallskip
\noindent \emph{Step 1. }
Suppose that $v < v'$ and $w = w'$.
To obtain a contradiction, suppose that $p > p'$.
Since both types share the same budget, the following two
incentive constraints are satisfied under the mechanism $\mu$:
	\begin{align*}
		v\,q \, - \, p 	\ \geq \ v\,q' \, - \, p'
			\qquad \text{and} \qquad
		v'q' \, - \, p' 	\ \geq \ v'q \, - \, p .
	\end{align*}
Re-arranging these expressions obtains
	\begin{equation*}
		v \big( q \, - \, q' \big)
		\ \geq \ p \, - \, p'
		\ \geq \ v' \big( q \, - \, q' \big) .
	\end{equation*}
Since we assumed that $p > p'$, from the first inequality
above one concludes that $q > q'$.
But this immediately implies that $v \geq v'$, which
is a contradiction.

\smallskip
\noindent \emph{Step 2. }
Suppose now that $v = v'$ and $w < w'$.
As in the previous step, to obtain a contradiction suppose
that $p > p'$.
Since the mechanism $\mu$ satisfies \eqref{Eq:BF}, it must
be that type $(v',w')$ can afford (ex-post) the menu entry
$(q,p)$ ---indeed, $p \leq w\,q \leq w'q$.
Thus, the \eqref{Eq:IC} constraint for type $(v',w')$
requires
	\begin{equation}  \label{Eq:rev-mon}
		v'q' \, - \, p' \ \geq \ v'q \, - \, p .
	\end{equation}
If \autoref{Eq:rev-mon} holds with equality, the buyer's type
$(v',w')$ is indifferent between lotteries $(q,p)$ and
$(q',p')$ but chose the latter one, despite the fact that
$p > p'$ and that both lotteries are ex-post budget feasible.
This contradicts the fact that $\mu$ is
seller-favorable.

Since \autoref{Eq:rev-mon} holds with strict inequality, 
noticing that $v' = v$ shows that lottery $(q',p')$
must not be ex-post budget feasible for $(v,w)$.
Otherwise, the buyer's type $(v,w)$ would strictly prefer to
purchase $(q',p')$ than $(q,p)$, which is a contradiction.
Thus, it must be that $w\,q \geq p > p' >  w\,q'$,
which implies $q > q'$.
Using \autoref{Eq:rev-mon} again (with $v$ instead of $v'$),
and rearranging obtains $v (q' - q) \, > \, p' - p
\, > \, w(q'-q)$, and thus
	\begin{equation*}
		(v \, - \, w) (q' \, - \, q) \ > \ 0 .
	\end{equation*}
	
Observe that we must have $v > w$, for otherwise
$v \leq w < p'/q'$, and thus $v\,q' - p' < 0$.
But then type $(v,w')$ is choosing a lottery with
negative utility, violating \eqref{Eq:IR}.
Readily, the above inequality implies $q' > q$, a contradiction
to our previous finding.

Of course, if $v < v'$ and $w < w'$, a combination of
the above steps gives us the result.
Thus, the pricing function $s$ is weakly monotone, as desired.
\end{proof}

\autoref{P:rev-mon} states that any mechanism $\mu$ in 
$\cali M$ satisfies a \emph{revenue monotonicity} property: 
for any two buyers $(V,W)$ and $(V',W')$ in the family 
$\scr B$, if $(V',W') \geq (V,W)$, then the expected revenue 
generated by the former buyer is weakly greater than 
the expected revenue generated by the latter one. 
It is clear that mechanisms in a superclass $\cali N' 
\supseteq \cali M$ need not satisfy revenue monotonicity, 
since $\cali N'$ may contain mechanisms that violate either 
\eqref{Eq:IC}, or \eqref{Eq:IR}, or \eqref{Eq:BF}. 
Here we highlight that the revenue monotonicity property 
depends on the mechanism $\mu$ being defined on $\R^2_+$.%
	\footnote{As the proof of \autoref{P:rev-mon} makes clear, 
		one can relax this to consider random variables with 
		a lattice as common support.}  
Because of the \emph{Extension Lemma}  
---see \autoref{L:extension} in \hyperlink{App-A}{Appendix A}---
treating $\R^2_+$ as the domain of mechanisms in $\cali M$ 
is without loss of generality.  
However, the \emph{Extension Lemma} may not apply to 
mechanisms in $\cali N \subseteq \cali M$: if $\mu \in 
\cali N$ is originally defined on a strict subset 
$D$ of $\R^2_+$, its extension $\overline\mu$ to $\R^2_+$, 
while still in $\cali M$, may violate the extra restrictions 
that define the subclass $\cali N$.
This implies that mechanisms in $\cali N$ defined on a 
strict subset of $\R^2_+$ need not satisfy 
the revenue monotonicity property ---such an environment  
is presented in \autoref{SS:sic}.  
Thus, an important consideration in using approximation 
mechanisms is whether revenue monotonicity can be preserved.  
As \citet{Rubinstein:2018fc} and \citet{BenMoshe:2024ej}, 
we measure this via revenue non-monotonicity gaps.

\begin{definition}  \label{D:RDM}
Let $\scr B$ be a family of $(V,W)$--buyers and let 
$\cali N$ be a non-empty class of mechanisms.  
Say that the \emph{revenue non-monotonicity gap} can be unboundedly 
large in $\cali N$ if there exists some $(V,W) \in \scr B$ 
for whom 
	\begin{equation*}
		\sup_{(V',W') \, \geq \, (V,W)}  \ 
		\frac{\Rev(V,W | \cali N)}{\Rev(V',W' | \cali N)} 
		\ = \ + \, \infty . 
	\end{equation*}
\end{definition}


\section[Revenue Gaps]{Revenue Gaps}  \label{S:Gaps}

In this section, we explore the gap between the revenue
generated by optimal mechanisms in $\cali M$ and the maximal
revenue generated by approximation mechanisms (that is,
mechanisms that lie in some class $\cali N \ne \cali M$).


\subsection[Good Approx]{Good Approximations with Simple
Mechanisms}   \label{SS:good-approx}

It is convenient to define the \emph{menu} generated by a 
mechanism $\mu = (x,s)$ as
	\begin{equation*}
		\Menu(\mu) \ := \
			\big\{ \big( x(v,w),s(v,w) \big) \, \in \,
			[0,1] \times \R_+ \, \colon \, (v,w) \in
			\R_+^2 \big\} \setminus \{(0,0)\} .
	\end{equation*}
In other words, the menu of $\mu = (x,s)$ is given
by its image, excluding the trivial lottery.%
		\footnote{Any ex-post individually rational mechanism
		can include the trivial lottery $(0,0)$.}
The \emph{menu size} of the mechanism $\mu$ is the
cardinality of its menu.
As before, we refer to the different components of
$\Menu(\mu)$ as lotteries or menu entries.
When the mechanism is fixed and no confusion arises, we
write generic lotteries as pairs $(q,p)$, where
$q \in [0,1]$ is the probability that the buyer gets the
object and $p \geq 0$ is the expected price associated with
this probability.
A lottery $(q,p)$ belongs to $\Menu(\mu)$ if there
exists some type $(v,w)$ for which $(x(v,w),s(v,w))
= (q,p)$.

For any positive integer $m$, let $\cali S_m \subseteq
\cali M$ denote the subclass of mechanisms that have
a menu size of (at most) $m$.
In a setting where the buyer faces no financial constraints,
\citet{Myerson1981} shows that for any random valuation $V$,
the menu size of the optimal mechanism is $m = 1$.
Thus, without budgets, there is no revenue gap between
the subclass $\cali S_1$ and $\cali M$ (where these two
sets are defined ignoring the \eqref{Eq:BF} constraint).

Unfortunately this result doesn't extend to a setting with
financial constraints.
Indeed, \citet{Che:2000ux} characterize the optimal selling
mechanism for a single item in a seller--buyer setting
with private valuation and private budget, where valuation
and budget may be correlated.
They show that the revenue maximizing mechanism may
require a continuum of lotteries.
This is true even in cases where the buyer has a publicly
known valuation.
As a result, any mechanism with finite menu size is suboptimal.
We adapt their example to our setting with ex-post budget
feasibility and participation constraints.

\smallskip

\begin{example}   \label{Ex:public-v}
The buyer has a publicly known valuation of $\hat v > 1$ and a
private budget that lies in the unit interval.  
The buyer's budget is uniformly distributed on $[0,1]$, 
which is common knowledge.
Thus, the support of the random vector $(V,W)$ is
$\{ \hat v \} \times [0,1]$.
In \hyperlink{App-B}{Appendix B}, we show that every optimal 
ex-post mechanism in this setting must have an infinite menu.
\hfill $\triangle$
\end{example}

\citet{Feng:2022qm} investigate the power of posted prices
as a selling mechanism in the presence of private budgets.
They showed that the revenue gap between any single
take-it-or-leave-it offer and the exact optimal mechanism
can be infinite, even for a single buyer setting
where the budget distribution and the valuation
distribution are independent from each other.
Therefore, the optimal revenue in their setting cannot
be well approximated using mechanisms in the
subclass $\cali S_1$.
In this subsection, we show that if the family of random
vectors contains only those with bounded support in terms
of both the valuation and the budget realizations, then
simple, finite menu size mechanisms provide a good
approximation to the optimal revenue.
In the next subsection we show how things change dramatically
when one considers random variables with unbounded support.
In these two cases, we do not assume independence between
the valuation and the budget.

Our first result on the revenue gap between optimal
mechanisms and optimal approximation mechanisms states that if
valuations and budgets are bounded above and below (away
from zero),
then an arbitrarily high fraction of the optimal revenue
can be extracted using a simple mechanism with finite
(polylogarithmic) menu size.
This result holds even if the family of buyers considered
by the seller includes random vectors where valuation
and budget are correlated.

\begin{proposition}  \label{P:good-approx}
Fix $0 < L < H < \infty$. 
Let $\scr B_b$ be a family of $(V,W)$--buyers that
contains all distributions whose supports are subsets of
$[L,H] \times [L,H]$.
Then, for every $\epsilon > 0$ there exists a positive
integer $m \, \leq \, (1  + \lfloor \log_{1+\epsilon}
H/L\rfloor )^2$ such that
	\begin{equation*}
		\GFOR ( \cali S_{m}; \scr B_b)
		\ \geq \ 1 \, - \, \epsilon .
	\end{equation*}
\end{proposition}

\begin{proof}
It is notationally convenient to rescale the supports 
of the $\scr B_b$ family to the bounded square 
$[1,H/L]\times[1,H/L]$.  
This won't have any effect on the ratio of revenues 
measured in GFOR.
Let now $(V,W)$ be a random vector in $\scr B_b$.
Fix any $\epsilon > 0$, and consider a $(V',W')$--buyer
such that
	\begin{equation*}
		(V',W') \ = \ \frac{1}{1+\epsilon} \, (V,W) .
	\end{equation*}
In words, type $(v',w') = (1+\epsilon)^{-1}(v,w)$ belongs to
the support of $(V',W')$ if and only if type $(v,w)$ is
in the support of $(V,W)$.
Clearly,
	\begin{equation} \label{Eq:good-approx}
		\Rev (V',W') \ = \ \frac{1}{1 \, + \, \epsilon}
			\, \Rev (V,W) ,
	\end{equation}
since $(x^*,s^*) \in \cali M$ is a revenue maximizing
mechanism for the $(V,W)$--buyer if and only if $(x^*,
s^*/(1+\epsilon)) \in \cali M$ is a revenue maximizing
mechanism for the scaled down $(V',W')$--buyer.

We construct a new random vector $(\hat V, \hat W)$
from $(V,W)$ as follows.
Replace every realization $(v,w)$ in the support of $(V,W)$
with $(\hat v, \hat w)$, where $\hat v$ is the highest
non-negative integer power of $1 + \epsilon$ that is smaller
than or equal to $v$, and likewise $\hat w$ is the highest
non-negative integer power of $1 + \epsilon$ that is smaller
than or equal to $w$.
Observe that different types in the support of $(V,W)$
may `collapse' to the same type in the rounded down
random vector $(\hat V,\hat W)$, in which case the
probability of type $(\hat v, \hat w)$ is the sum of
probabilities of the collapsing types in $(V,W)$.
By construction, the $(\hat V, \hat W)$--buyer contains
at most
	\begin{equation}  \label{Eq:m*}
		m^*
		\ = \ (1 \, + \, \lfloor \log_{1+\epsilon}
			H/L \rfloor)^2 	
	\end{equation}
distinct types. 
To see this, notice that $\hat v \in \{ (1+\epsilon)^0, 
(1+\epsilon)^1,(1+\epsilon)^2,	\ldots, 
(1+\epsilon)^{\lfloor \log_{1+\epsilon}	H/L \rfloor} 
\}$, and likewise for $\hat w$. 
Thus, the optimal mechanism in $\cali M$ for
the $(\hat V, \hat W)$--buyer has menu size $m \leq m^*$.
Notice that for every type $(v,w)$ in the support of $(V,W)$,
we have $v/(1+\epsilon) \leq \hat v \leq v$ and
$w/(1+\epsilon) \leq \hat w \leq w$, where $(\hat v, \hat w)$
is the corresponding `collapsed' type in the support
of $(\hat V, \hat W)$.

\citet{Babaioff:2022by} show that any mechanism can be 
transformed into a seller-favorable mechanism by 
multiplying the price of each menu by $(1-\epsilon')$, 
for sufficiently small $\epsilon' > 0$.  
The resulting mechanism achieves revenue of at least 
a $(1-\epsilon')$-fraction of the original mechanism's 
revenue. 
Hence, there exists a seller-favorable feasible mechanism 
$\mu'$ for $(V',W')$ such that 
	\begin{equation}
		(1 - \epsilon') \Rev(V',W')  
		\ \leq \ R(\mu';V',W') .   \label{Eq:fix-a}
	\end{equation}
Since $(V',W') \leq (\hat V, \hat W)$ and $\mu'$ is 
seller-favorable, \autoref{P:rev-mon} implies 
	\begin{equation}
		R(\mu';V',W') \ \leq \ R(\mu';\hat V, \hat W) .
			\label{Eq:fix-b}
	\end{equation}

Because $(\hat V, \hat W)$ has finite support with at most 
$m^*$ types, there exists a seller-favorable mechanism 
$\hat\mu$ for $(\hat V, \hat W)$ whose menu size is at most 
$m^*$ and moreover this mechanism is optimal among all 
seller-favorable mechanisms for type $(\hat V, \hat W)$. 
Clearly, 
	\begin{equation}
		R(\mu';\hat V, \hat W) 
		\ \leq \ R(\hat\mu;\hat V, \hat W) . 
			\label{Eq:fix-c}
	\end{equation}
Also, $(\hat V, \hat W) \leq (V,W)$, so \autoref{P:rev-mon} 
shows that 
	\begin{equation}
		R(\hat\mu; \hat V, \hat W) 
		\ \leq \ R(\hat\mu; V,W) .
			\label{Eq:fix-d}
	\end{equation}
Because $\hat\mu$ has menu size at most $m^*$, we have 
	\begin{equation}
		R(\hat\mu;V,W) \ \leq \ 
		\Rev(V,W | \cali S_{m^*}) . 
			\label{Eq:fix-e} 
	\end{equation}

Combining Equations \eqref{Eq:fix-a} to \eqref{Eq:fix-e} 
and using \autoref{Eq:good-approx} obtains
	\begin{equation*}
		\frac{1 \, - \, \epsilon'}{1 \, + \, \epsilon}
		\Rev(V,W) 
		\ = \ (1 \, - \, \epsilon')\Rev(V',W') 
		\ \leq \ \Rev (V,W | \cali S_{m^*}) ,  
	\end{equation*}
or equivalently 
	\begin{equation*}
		\frac{1 \, - \, \epsilon'}{1 \, + \, \epsilon}
		\ \leq \ \frac{\Rev (V,W | \cali S_{m^*})}
			{\Rev (V,W)} .
	\end{equation*}
Clearly, for sufficiently small $\epsilon' > 0$ we have 
that $(1-\epsilon) \leq (1-\epsilon')/(1+\epsilon)$. 
Our argument holds for any $(V,W)$--buyer in the
class of distributions $\scr B_b$ with support inside the
bounded rectangle $[1,H/L] \times [1,H/L]$.
Thus, we conclude that
	\begin{equation*}
		1 \, - \, \epsilon
		\ \leq \ \GFOR ( \cali S_m ; \scr B_{b} )  
		\qedhere
	\end{equation*}
\end{proof}

\smallskip

\begin{remark}
\citet{Hart:2019ny} introduced the `nudge-and-round' 
technique, which subtly discretizes the menu of a given 
mechanism. 
In contrast, our approach discretizes the valuation and the 
budget, and employs a revenue monotonicity argument to 
establish approximation guarantees. 
Notably, revenue monotonicity does not hold in 
\citet{Hart:2019ny}'s setting, which involves 
multiple goods (cf.~\citet{Hart:2019it}), 
whereas our single-item, budget-constrained environment 
admits such a property --- see \autoref{P:rev-mon}.
\end{remark}

\smallskip

\begin{remark}
Notice that we can bound $m^*$, which itself is an upper
bound for the menu size complexity needed in
\autoref{P:good-approx}, by the following:
	\begin{align*}
		m^* \ = \
		(1 + \lfloor \log_{1+\epsilon} H/L \rfloor)^2
		& \ \leq \ 4 \, ( \log_{1+\epsilon} 
			H/L )^2   \\[.75ex]
		& \ \leq \ 4 \, \left( 1 \, + \, \tfrac{1}{\epsilon}
			\right)^2 \, (\log_{2} H/L )^2 ,
	\end{align*}
where the last inequality is obtained by changing the base
of the logarithm to $2$ and using the fact that
$1/\log_2 (1+\epsilon) \, \leq \, 1 + 1/\epsilon$, for every
$\epsilon > 0$.%
	\footnote{Using the Big O notation and the fact that 
	$(1 + \frac{1}{\epsilon})^2 \leq \frac{4}{\epsilon^2}$
	for every $0 < \epsilon < 1$, 
	we have that $m^* = O \big( ( \frac{\log_2 H/L}
	{\epsilon})^2 \big)$. }
Consequently, the menu size depends only on $\epsilon$ 
and $H/L$ in a polylogarithmic way and is independent 
of the distribution. 
It has a reasonable size even for large upper
bounds for the valuation and/or the budget.
For fixed bounds on the support of the distributions
considered, a lower $\epsilon > 0$ (equivalently, a better
approximation to the optimal revenue) demands a larger
menu size complexity.
For instance when $L = 1$ and $H = 4$, \autoref{Eq:m*} shows 
that for $\epsilon = 1/2$ (hence, to approximate half the 
optimal revenue) it suffices to consider mechanisms with menu 
size of at most $16$.
But the upper bound on menu size complexity increases to
$64$ if $\epsilon = 1/5$.
\end{remark}


\subsection[Bad Approx]{Bad Approximations with Simple
Mechanisms}  \label{SS:bad-approx}

A conjecture that seems to be suggested from 
\autoref{P:good-approx} is that, as the common upper bound $H$ 
of the support of the random vector $(V,W)$ increases, 
the menu size complexity required to guarantee an
arbitrarily high fraction of the optimal revenue tends
towards infinity.
The validity of this conjecture is not immediately
verified for three reasons.
First, \autoref{P:good-approx} provides a polylogarithmic
upper bound on the menu size complexity, but it
does not provide a matching lower bound.
Thus, we cannot be sure of how tight the upper bound is.
Second, the right-hand side of \autoref{Eq:m*}
grows in a sublinear way.%
		\footnote{We discuss this concept after the
		proof of \autoref{P:bad-approx-bounded}.}
So, it is possible that, as $H$ grows, letting the size of 
the menu grow faster (i.e., linearly) would keep the good 
approximation result in place.
Finally, it is not clear what happens when the support
of the marginal distribution of $V$ is bounded and
that of the marginal distribution of $W$ is unbounded
(or vice versa).

However, our next result shows that despite these
considerations, the conjecture holds.
If the menu size of the approximation mechanisms grows
sublinearly when the upper bound $H$ tends to infinity, 
the revenue loss can become unboundedly large.

\begin{proposition}  \label{P:menu-size}
Let $\scr B_{u}$ be a family of $(V,W)$--buyers that
contains all distributions with supports that are
unbounded from above and bounded away from zero from below.
Then for any fixed positive integer $m$,
	\begin{equation*}
		\GFOR (\cali S_m; \scr B_{u}) \ = \ 0 .
	\end{equation*}
\end{proposition}

\begin{proof}
Fix $k \geq 2$ and let $B = 2k$.
Consider a $(V,W)$--buyer with a public valuation
and $k$ private budgets.
More specifically, the support for the $(V,W)$--buyer is
given by
	\begin{equation}   \label{Eq:supp-menu-size}
		\big\{ (v_i,w_i) \,\colon\,
			\text{$w_i \, = \, B^i$, for all $i = 1,\ldots,k$,
			and $v_i = \overline v \, \geq \, B^{k+1}$}  \big\} .
	\end{equation}
Further, assume that the probability mass function $f$
for the random vector $(V,W)$ is
	\begin{equation}  \label{Eq:prob-menu-size}
		f(v_i,w_i) \ \equiv \ f_i \ = \
		\begin{cases}
			(1 \, - \, B^{-1}) \, B^{-(i-1)}
		 		& \, \colon \, \text{$i \, = \, 1,\ldots,k-1,$
			 	} \\[1ex] 	
			 1 \, - \, \sum\limits_{j=1}^{k-1} f_j
				& \, \colon \,  \text{$i = k$.}
		\end{cases}
	\end{equation}

The proof of \autoref{P:menu-size} follows from two
lemmas that we make in regard to the revenues that can
be raised from the $(V,W)$--buyer.

\begin{lemma}  \label{L:bound-1}
A lower bound for the optimal revenue from the $(V,W)$--buyer 
is $B^2/8$; i.e., 
	\begin{equation*}
		\Rev(V,W) \ \geq \ B^2/8  .
	\end{equation*}
\end{lemma}  
	
\begin{proof}
Consider a mechanism $\mu = (x,s)$ with
$\Menu(\mu)  = \{ (q_i,p_i) \colon i = 1,\ldots,k \}$ given by
	\begin{align*}
		q_i
		\ = \ \frac{1}{2} \, + \, \frac{i}{B}	
			\qquad \text{and} \qquad
		p_i \ = \ q_i \, B^i ,
			\qquad \quad \text{for all $i = 1,\ldots,k$.}
	\end{align*}
This mechanism is clearly ex-post budget feasible and
ex-post individually rational.
Upon choosing the lottery $(q_i,p_i)$, the buyer with type
$(v_i,w_i)$ pays nothing if the good is not allocated and
exhausts her budget if the good is allocated ---since the
price she actually pays in this case is $p_i/q_i = B^i = w_i$.
Since $v_i = \overline v \geq w_i$ for all $i = 1,\ldots,k$, the
mechanism $\mu = (x,s)$ indeed satisfies both \eqref{Eq:BF}
and \eqref{Eq:IR}.
Note also that $1/2 < q_1 < q_2 < \ldots < q_k = 1$.

We now show that $\mu = (x,s)$ also satisfies
\eqref{Eq:IC}.
Since the buyer's type $(v_i,w_i)$ cannot afford to buy
menu entry $(q_{h},p_{h})$, for any $h > i$, it suffices
to verify  that $(v_i,w_i)$ does not gain by choosing
any lottery $(q_j,p_j)$, for any $1 \leq j < i$.
We claim that
	\begin{equation}  \label{Eq:rev-gap}
		\overline v \, (q_i \, - \, q_{i-1})
		\ \geq \ p_i \, - \, p_{i-1}
	\end{equation}
holds for all $i = 2,\ldots,k$.
Clearly, this implies that $\overline v\,q_i - p_i \geq
\overline v \,q_{j} - p_{j}$, for all $1 \leq j < i$.
To verify \autoref{Eq:rev-gap}, note
	\begin{align*}
		\overline v \, (q_i \, - \, q_{i-1})
		\ - \ (p_i \, - \, p_{i-1})
		& \ \geq \ \overline v \, (q_i \, - \, q_{i-1}) \ - \
			p_i  \\[.75ex]
		& \ = \ \overline v \, B^{-1}  \, - \, q_i\,B^i
		  \ \geq \ \overline v \, B^{-1} \, - \, B^{i}
		  \ \geq \ 0 ,
	\end{align*}
where the second inequality follows from $0 < q_i \leq 1$
and the last from the fact that $\overline v \geq B^{k+1}$.
Thus, the mechanism $\mu = (x,s)$ satisfies \eqref{Eq:IC}.
We conclude that $\mu \in \cali M$.

The expected revenue generated by $\mu = (x,s)$ is
$R(\mu; V,W) \, = \, \sum_{i=1}^{k-1} p_i f_i \, + \, p_k f_k$.
The first term of this expression can be written as
	\begin{align*}
		\sum_{i=1}^{k-1} p_i \, f_i
		& \ = \ \sum_{i=1}^{k-1} q_i \, B^i \, (1 - B^{-1})
			B^{-(i-1)}  \\[1ex]
		& \ = \ \sum_{i=1}^{k-1} \left( \tfrac{1}{2} \, + \,
			\tfrac{i}{B} \right) \big( B - 1 \big)  \\[1ex]
		& \ = \ \tfrac{1}{2} \left(B \, - \, 1\right)
			\left(k \, - \, 1 \right)
			\left( 1 \, + \, \tfrac{k}{B} \right)
		  \ = \ \tfrac{3}{4} \left(B \, - \, 1\right)
		  	\left( \tfrac{B}{2} \, - \, 1 \right) ,
	\end{align*}
where the last equality follows from the fact that
$k = B/2$.
For the second term in $R(\mu;V,W)$, notice that
	\begin{align*}
		f_k
		\ = \ 1 \, - \, \sum_{j=1}^{k-1} f_j
		\ = \ 1 \, - \, \sum_{j=1}^{k-1} \big(1 -
			B^{-1} \big) \, B^{-(j-1)}
		  \ = \ B^{-(k-1)} ,
	\end{align*}
where the last equality uses the identity for the geometric
sum.
Since $p_k = B^k$, expected revenue coming from the highest
budget buyer is just $B$.

Putting these two observations together obtains
	\begin{align*}
		R(\mu;V,W)
		& \ = \ \tfrac{3}{4} \left(B \, - \, 1\right)
		  	\left( \tfrac{B}{2} \, - \, 1 \right) \ + \ B
		  		\\[1ex]
		& \ \geq \ \tfrac{3}{4} \left( \tfrac{B}{2} \, - \, 1
			\right)^2 \ + \ \tfrac{3}{4}B
		 \ = \ \tfrac{3}{4} \left( \tfrac{B^2}{4} \, + \,
			1  \right)
			\ \geq \ \tfrac{B^2}{8} .
	\end{align*}
Since $\Rev(V,W) \geq R(\mu;V,W)$, we have that $B^2/8$
is a lower bound on the expected revenue coming from
the optimal mechanism.   
\end{proof}

\begin{lemma}   \label{L:bound-2}
For any $1 \leq m < k$, an upper bound for the optimal 
revenue from the $(V,W)$--buyer in the class of mechanisms 
with menu size $m$ is $(m+1)B$; i.e., 
	\begin{equation*}
		\Rev (V,W | \cali S_m) \ \leq \ (m \, + \, 1)\,B .
	\end{equation*}
\end{lemma} 

\begin{proof}
To show this, let $\nu$ be any direct mechanism with menu
size $m$.
For positive integers $1 \leq h_1 < h_2 < \ldots < h_{m}
\leq k$, write $\Menu(\nu) = \{ (q_{h_i},p_{h_i}) \colon i
= 1,\ldots,m\}$ to denote its menu.
Disregarding the incentive constraints, since we are looking
for an upper bound, we consider a mechanism $\nu$ for which
$(q_{h_i},p_{h_i}) = (1,w_{h_i}) = (1,B^{h_i})$, for all
$i = 1,\ldots,m$.

The buyer with type $(v_{h_i},w_{h_i})$ finds the lottery
$(q_{h_i},p_{h_i})$ ex-post affordable and ex-post
individually rational, but so do all types $h_i \leq j
< h_{i+1}$ ---recall that we are not considering incentive
compatibility, and thus the mechanism will direct the
type $(v_{h_{i+1}},w_{h_{i+1}})$ to purchase menu entry
$(q_{h_{i+1}},p_{h_{i+1}})$, and so on.
Therefore, the expected revenue generated by this mechanism
is
	\begin{equation}  \label{Eq:rev-m-size}
		\sum_{h_1 \leq j < h_2} \!\!  B^{h_1} f_j \ + \
		\sum_{h_2 \leq j < h_3} \!\!  B^{h_2} f_j \ + \
			\ldots \ + \
		\sum_{h_m \leq j < k} \!\! B^{h_m} f_j \ + \
			B^{h_m} f_k .
	\end{equation}

Focusing on the first $m$ terms of the above expression,
we write $n_i$ to denote the number of terms in
the $i$-th sum.
That is, there are $n_i$ terms between $h_i$ and $h_{i+1}$,
including the first one but omitting the second one, for
each $i = 1,\ldots,m$.
With this notation, we can write
	\begin{align*}
		\sum_{h_{i} \leq j < h_{i+1}} \!\! B^{h_i}f_j 	
		& \ = \ \big( 1 \, - \, B^{-1} \big) \, B^{h_i}
			\sum_{h_{i} \leq j < h_{i+1}} \!\! B^{-(j-1)}
				\\[1ex]
		& \ = \ \big( 1 \, - \, B^{-1} \big) \, B^{h_i} \,
			B^{-(h_i - 1)} \sum_{j=1}^{n_i} \, B^{-(j-1)}
				\\[1ex]
		& \ = \ \big( 1 \, - \, B^{-1} \big) \, B \
			\frac{1 \, - \, B^{- n_i}}{1 - B^{-1}}
		  \ = \ B \big( 1 \, - \, B^{-n_i} \big)
		  \ \leq \ B .
	\end{align*}
The final term of the expected revenue in
\autoref{Eq:rev-m-size} is
	\begin{align*}
		B^{h_m} f_k
		\ = \ B^{h_m} \, B^{-(k-1)} \ \leq \
			B^k \, B^{-(k-1)} \ = \ B.
	\end{align*}

It follows that the expected revenue generated by the
mechanism $\nu$ with menu size $m$ is at most $(m + 1) B$.
Note that this doesn't depend on the choices for
$h_1,h_2,\ldots,h_m$.
Thus, we conclude that any mechanism with menu size
$m < k$ that satisfies \eqref{Eq:IR} and
\eqref{Eq:BF}, generates a revenue of at most $(m + 1) B$.
Thus, we can take this last number as an upper bound
for the revenue generated by any mechanism with
the same menu-size complexity that in addition satisfies
\eqref{Eq:IC}.
Hence, $\Rev (V,W | \cali S_m) \leq (m + 1) B$.  
\end{proof}

We now put things together.
From Lemma~\ref{L:bound-1}, we know that $\Rev (V,W) \geq B^2/8$.
Lemma~\ref{L:bound-2} shows that $\Rev (V,W | \cali S_m)
\leq (m + 1) B$.
Thus, since $B = 2k$, the following ratio obtains
	\begin{equation}  \label{Eq:ratio-Sm}
		\frac{\Rev (V,W | \cali S_m)}{\Rev (V,W)}
		\ \leq \ \frac{8(m+1)}{B}
		\ = \ \frac{4(m+1)}{k} .
	\end{equation}
Finally, since we are looking at the family of buyer's
distributions with unbounded support for $W$, letting
$k \to \infty$ shows that $\GFOR (\cali S_m; \scr B_{u})
\, = \, 0$, as desired.   
\end{proof}

It would appear that the fact that $\scr B_u$ contains
random vectors $(V,W)$ with unbounded support in both
$V$ and $W$ is crucial for our bad approximation result
in \autoref{P:menu-size}.
Indeed, if either the marginal distribution of $V$ or $W$
has bounded support, then $\Rev(V,W) \leq
\mathbb{E}\big[ \min \{V,W\} \big] < + \infty$.%
		\footnote{For any direct mechanism $\mu = (x,s)
		\in \cali M$, one has $s(v,w) \leq \min\{v,w\}\,x(v,w)
		\leq \min\{v,w\}$ by the ex-post budget feasibility
		and individual rationality constraints. }
But this is not correct.
The key to the above result is the observation that,
while both $\Rev (V,W)$ and $\Rev (V,W | \cali S_m)$
may increase without bound, the first one does so at a
faster rate.
This can happen as well when the lower bounds for
the valuation and budget become arbitrarily close to zero.

In particular, we can show the following.

\begin{proposition}  \label{P:bad-approx-bounded}
Let $\scr B_{0}$ be a family of $(V,W)$--buyers that
contains all distributions whose supports are subsets
of $[0,1] \times [0,1]$.
Then for any fixed positive integer $m$,
	\begin{equation*}
		\GFOR (\cali S_m; \scr B_{0}) \ = \ 0 .
	\end{equation*}
\end{proposition}

\begin{proof}
Fix $k \geq 2$ and $B = 2k$.
Consider now a $(\hat V,\hat W)$--buyer with support given by
	\begin{equation*}
		\big\{ (\hat v_i,\hat w_i) \,\colon\,
		\text{$\hat v_i \, = \, 1$ and $\hat w_i \, = \,
		1/B^{B-i}$, 	for all $i = 1,\ldots,k$}  \big\}
	\end{equation*}
and probability mass function $f$ given by
\autoref{Eq:prob-menu-size}.
Notice that this $(\hat V, \hat W)$--buyer in the family
$\scr B_0$ has been obtained from the $(V,W)$--buyer
in the proof of \autoref{P:menu-size} by normalizing the
valuation $\overline v$ to $1$ and dividing each budget
$w_i$ by $B^B$ (see \autoref{Eq:supp-menu-size}).

Readily modifying Lemma~\ref{L:bound-1}  
obtains $\Rev (\hat V, \hat W) \, \geq \, B^2/8 B^B$.  
Likewise, a simple modification of Lemma~\ref{L:bound-2} 
obtains the following inequality, which holds for any fixed 
integer $1 \leq m 
< k$: $\Rev (\hat V, \hat W | \cali S_{m}) \, \leq 
\, (m+1)B / B^{B}$.  
Using the fact that $B = 2k$, we can write again
	\begin{equation*}
		\frac{\Rev (\hat V, \hat W | \cali S_m)}
			{\Rev (\hat V, \hat W)} 
		\ \leq \ \frac{8(m+1)}{B} \ = \ \frac{4(m+1)}{k} . 
	\end{equation*}
Finally, letting $k \to \infty$ shows that $\GFOR 
(\cali S_m; \scr B_{0}) \, = \, 0$.   
\end{proof}

A closer look at the proof of \autoref{P:menu-size} 
enables us to strengthen this result. 
Indeed, the revenue gap between the optimal mechanism
and the optimal approximation mechanism can be
exceedingly large as long as the menu size complexity
of the class of approximation mechanisms we consider
grows in a sublinear way.
Recall that a positive real-valued function $\phi$
defined on the set of positive integers is called
\emph{sublinear} if $\lim_{k \to \infty} \phi(k)/k = 0$.
This is often written as $\phi \in o(k)$, where 
$o(k)$ refers to the class of functions that grow
asymptotically slower than any linear function.  
Thus, we can show that any mechanism with asymptotically
sublinear menu size complexity cannot guarantee a positive
fraction of the optimal revenue, if the buyer's family of
random-vectors is either $\scr B_u$ or $\scr B_0$.

\begin{proposition}  \label{P:sublinear}
Let $\{(V_k,W_k) \colon k = 1,2,\ldots \}$ be a sequence of
random vectors, where for each $k$ the distribution of
$(V_k,W_k)$ is given in \autoref{Eq:prob-menu-size}.
Then, for any sublinear function $\phi$,
	\begin{equation*}
		\lim_{k \to \infty} \frac{\Rev (V_k,W_k |
			\cali S_{\phi(k)})}{\Rev(V_k,W_k)}
		\ = \ 0 .
	\end{equation*}
\end{proposition}

\begin{proof}
Immediately by substituting $m = \phi(k)$ in 
\autoref{Eq:ratio-Sm}. 
Clearly, if $\phi \in o(k)$ then $\lim 4( \phi(k) + 1)/k 
= 0$ when $k \to \infty$.  
\end{proof}

\begin{remark}
While the marginal distribution of the budgets in the proofs of
\autoref{P:menu-size} and \autoref{P:bad-approx-bounded}
is important in our construction,
this distribution shares similar properties with the
standard geometric distribution.
In particular, it gives diminishing weight to types with
higher budgets.
\end{remark}

Our results in this subsection can be roughly summarized 
as follows: for certain families of buyers, given any 
$m \in \mathbb{N}$ there exists a buyer's distribution 
with \emph{finite support} for which only 
an arbitrarily small fraction of the optimal revenue can be 
extracted by a mechanism of menu size $m$. 
One can of course ask what happens with revenue gaps generated 
by simple mechanisms when one considers distributions with 
\emph{infinite support}.  
As it turns out, one can construct distributions with 
infinite support for which no finite menu size mechanism 
can guarantee any positive fraction of the optimal revenue. 
Consequently, in these cases any non-trivial approximation 
mechanism requires infinitely many lotteries.

\begin{proposition}  \label{P:infinite-supp} 
There exist (i) a correlated random vector $(V',W')$ with 
infinite support, and (ii) an independent random vector 
$(V'',W'')$ with infinite support, such that no feasible 
mechanism with finite menu size can guarantee any 
positive fraction of the optimal revenue. 
\end{proposition}

\begin{proof}
We divide the proof in three lemmas. 
The first one provides an upper bound on the optimal 
revenue that can be captured from any $(V,W)$--buyer using 
selling mechanisms with finite menu size.

\begin{lemma}   \label{L:bounds-m-class}
Let $(V,W)$ be a (correlated or independent) random vector,  
and let $Z = \min\{V,W\}$.  
Then, for any positive integer $m$, one has 
	\begin{equation}  \label{Eq:bounds-m-class}
		\Rev(V,W | \cali S_m) 
		\ \leq \ m \, \sup_{p \geq 0} \, p \, \Prob 
			\{Z \geq p\} .
	\end{equation}
\end{lemma}

\begin{proof}
Consider a mechanism $\mu \in \cali S_m$ with 
$\Menu(\mu) = \{(q_1,p_1),\ldots,(q_m,p_m)\}$. 
Since $\cali S_m \subseteq \cali M$, for any 
type realization $(v,w)$ we have $p \leq v \, q \leq v$
and $p \leq w\,q \leq w$.  
With this simple observation, we can bound the revenue 
generated by $\mu$: 
	\begin{align*}
		R(\mu;V,W) 
		\ \leq \ \sum_{i=1}^m \, p_i \, \Prob \{Z \geq p_i\} 
		\ \leq \ m \, \sup_{p \geq 0} \, p \, \Prob 
			\{Z \geq p\}  .
	\end{align*}
Since this holds for any $\mu$ in $\cali S_m$, the result 
in \autoref{Eq:bounds-m-class} is obtained. 
\end{proof}

Using \autoref{L:bounds-m-class}, to prove 
\autoref{P:infinite-supp} it suffices to find 
a correlated random vector $(V',W')$ and an independent 
random vector $(V'',W'')$ for which the optimal revenue
in $\cali M$ is infinite and, at the same time, 
	\begin{align}
		\sup_{p \geq 0} \, p \, \Prob \{ \min(V',W') \geq p\} 
		& \ < \ \infty ,
		\quad \text{and}   \label{Eq:6-fix-1} \\
		\sup_{p \geq 0} \, p \, \Prob \{ \min(V'',W'') \geq p\} 
		& \ < \ \infty .  \label{Eq:6-fix-2}
	\end{align}

\begin{lemma}
Let $(V',W')$ be a correlated random vector with support  
	\begin{equation*}
		\big\{ (v_i,w_i) \,\colon\, 
			\text{$v_i \, = \, 2^{2i}$ 
			\ and \ $w_i \, = \, 2^i$, 
			\ for all $i = 1,2,\ldots$}  \big\} 
	\end{equation*} 
and a probability mass function $f$ given by 
	\begin{equation*}
		f(v_i,w_i) \ \equiv \ f_i \ = \ 1/2^i , 
			\qquad \text{for all $i = 1,2,\ldots$.}
	\end{equation*} 
The optimal revenue from the $(V',W')$--buyer is 
$\Rev(V',W') = \infty$ while \autoref{Eq:6-fix-1} holds. 
\end{lemma}

\begin{proof}
Consider a mechanism $\mu = (x,s)$ defined by 
	\begin{equation*}
		q_i \, \equiv \, x(v_i,w_i) \, = \, 1 \, - \, 2^{-i} 
		\qquad \text{and} \qquad 
		p_i \, \equiv \, s(v_i,w_i) \, = \, (1 \, - \, 2^{-i}) 
			\, 2^i . 
	\end{equation*}
Note that  the mechanism $\mu$ has an infinite number of 
menu options $i = 1,2,\ldots$, each of which extracts the 
entire budget of type $(v_i,w_i)$ in case of winning the 
item, since by construction $p_i/q_i = w_i$.  
This prevents `upward' deviations to a lottery $k > i$.  

For any $j < i$, the difference in expected utility for type 
$i$ when choosing options $i$ and $j$ is given by 
	\begin{align*}
		(v_i\,q_i \, - \, p_i) \ - \ (v_i\,q_j \, - \, p_j) 
		& \ = \ v_i (q_i \, - \, q_j) \ - \ 
			(p_i \, - \, p_j) \\
		& \ = \ v_i(2^{-j} \, - \, 2^{-i}) \ - \ 
			(2^i \, - \, 2^j) \\
		& \ > \ 2^{2i}/2^i \ - \ 2^i \ = \ 0 . 
	\end{align*}
This shows that `downward' deviations are not profitable. 
It is immediate to verify that $v_iq_i - p_i \geq 0$ 
for all $i = 1,2,\ldots$, and thus we conclude that the  
mechanism $\mu \in \cali M$. 
Now, the total expected revenue generated by $\cali M$ is 
	\begin{equation*}
		R(\mu;V',W') 
		\ = \ \sum_{i=1}^{\infty} \, p_i \, f_i 
		\ = \ \sum_{i=1}^{\infty} \, (1 \, - \, 
			2^{-i}) \ = \ \infty .
	\end{equation*}
Of course this shows $\Rev(V',W') = \infty$, as desired. 

Clearly, $Z = \min(V',W') = W'$ with probability mass 
$f_i = 2^{-i}$.
For a posted price $p = 2^n$, one has that 
	\begin{equation*}
		\Prob \{ Z \, \geq \, 2^n \} 
		\ = \ \sum_{i=n}^\infty 2^{-i} 
		\ = \ 2^{-n} \Bigl( \sum_{j=0}^\infty 2^{-j} \Bigr) 
		\ = \ 2^{-n} \ 2  \ = \ 2^{-(n-1)} .
	\end{equation*}
It follows that the posted price $p = 2^n$ generates 
expected revenue equal to $p \, \Prob \{ Z \geq p\} 
\, = \, 2^n \, 2^{-(n-1)} \, = \, 2$. 
Since the expected revenue raised by any price that is a 
power of $2$ is exactly equal to $2$, we conclude that 
$\sup_{p \geq 0} p \, \Prob \{ Z \geq p \} = 2 < \infty$; 
i.e., \autoref{Eq:6-fix-1} holds. 
\end{proof}

\begin{lemma}
Let $(V'',W'')$ be an independent random vector with support 
	\begin{equation*}
		\big\{ (v_i,w_i) \,\colon\, 
			\text{$v_i \, = \, 2^{2i}$ 
			\ and \ $w_i \, = \, 2^i$, 
			\ for all $i = 1,2,\ldots$}  \big\} ,
	\end{equation*} 
where $V''$ realizes valuation $2^{2i}$ with probability 
$1/i(i+1)$ and $W''$ realizes budget $2^i$ with probability 
$2^{-i}$, for all $i = 1,2,\ldots$.  
The optimal revenue from the $(V'',W'')$--buyer is 
$\Rev(V'',W'') = \infty$ while \autoref{Eq:6-fix-2} 
holds. 
\end{lemma}

\begin{proof}
Consider a mechanism with infinite menu size 
whose option $k = 1,2,\ldots$ provides the good with 
probability $q_k = 1 - 2^{-k}$ for an expected price 
equal to $p_k = (1-2^{-k}) 2^k$. 
This mechanism is designed so that, for every type realization
$(v_i,w_j)$, the optimal choice is menu entry 
$k(v_i,w_j) = \min \{i,j\}$.  
Being so, this mechanism satisfies \eqref{Eq:IR}, 
\eqref{Eq:BF} and \eqref{Eq:IC}, and thus it 
belongs to $\cali M$.   

Indeed, the expected utility of type $(v_i,w_j)$ from 
choosing entry $k(v_i,w_j)$ is 
	\begin{align*}
		v_i q_{k(v_i,w_j)} \ - \ p_{k(v_i,w_j)} 
		\ = \ 2^{2i} \big(1 - 2^{- k(v_i,w_j)}\big) 
			\ - \ \big(1 - 2^{- k(v_i,w_j)} \big) 
			2^{k(v_i,w_j)} 
		\ \geq \ 0 ,
	\end{align*}
since $i \geq k(v_i,w_j)$.  
This shows that \eqref{Eq:IR} is satisfied.  
In addition, note that $p_{k(v_i,w_j)}/ q_{k(v_i,w_j)} 
\, = \, 2^{k(v_i,w_j)} \, \leq \, 2^j = w_j$, since 
$j \geq k(v_i,w_j)$.
This shows that \eqref{Eq:BF} is also satisfied.

To show \eqref{Eq:IC}, we first show that the menu entry 
$k(v_i,w_j)$ maximizes expected utility for type realization 
$(v_i,w_j)$ among all options with a smaller index, and 
also among all options with a higher index.  
To see that there is no profitable `downward deviation', 
consider any $k < k(v_i,w_j)$.  
Type $(v_i,w_j)$'s valuation difference from selecting 
menu $k(v_i,w_j)$ instead of menu $k$ is  
	\begin{align*}
		v_i \, \big( q_{k(v_i,w_j)} \, - \, q_k \big) 
		\ = \ v_i \big( 1 - 2^{-k(v_i,w_j)} \, - \, 
			1 \, + \, 2^{-k} \big)  
		\ = \ v_i \big( 2^{-k} \, - \, 
			2^{-k(v_i,w_j)} \big)  \ \geq \ 
			\frac{2^{2i}}{2^{k(v_i,w_j)}}  .
	\end{align*}
Likewise, the price difference between these two menu 
options is 
	\begin{align*}
		p_{k(v_i,w_j)} \, - \, p_k 
		\ = \ \big(1 - 2^{-k(v_i,w_j)}\big)2^{k(v_i,w_j)}
			\, - \, (1 - 2^{-k})2^{k}  
		\ = \ 2^{k(v_i,w_j)} \, - \, 2^k 
			\ < \ 2^{k(v_i,w_j)} . 
	\end{align*}
Putting these two observations together obtains 
	\begin{align*}
		v_i \big( q_{k(v_i,w_j)} \, - \, q_k \big) 
			\ \geq \ \frac{2^{2i}}{2^{k(v_i,w_j)}} 
			\ > \ 2^{k(v_i,w_j)} 
			\ > \ p_{k(v_i,w_j)} \, - \, p_k . 
	\end{align*}
	
To see that there is no profitable `upward deviation', 
we consider two sub-cases.  

When $k(v_i,w_j) = j$, it suffices to show that the 
$j$ menu entry is the only budget feasible menu 
entry for type realization $(v_i,w_j)$ among all 
menu options with an index higher than $j$.  
Observe that $p_{k(v_i,w_j)}/q_{k(v_i,w_j)} = 
2^{k(v_i,w_j)} = 2^j = w_j$, thus the menu entry $j$ 
exhausts type $(v_i,w_j)$'s budget in case of winning 
the object.  
Any higher indexed menu is unaffordable.

When $k(v_i,w_j) = i$, it suffices to show that menu entry 
$i$ maximizes type $(v_i,w_j)$'s expected utility among all 
options with an index higher than $i$.  
Consider any $h > i$; we claim that $v_i q_i - p_i 
> v_i q_h - p_h$.  
The left- and right-hand sides of this inequality are, 
in that order,  
	\begin{align*}
		\big(2^{2i} \, - \,2^i\big)\,
		\big(1 \, - \, 2^{-i}\big) 
		& \ = \ 2^{2i} \, - \, 2^{i+1} \, + \, 1 , 
		\qquad \text{and}  \\ 
		\big(2^{2i} \, - \,2^h\big)\,
		\big(1 \, - \, 2^{-h}\big) 
		& \ = \ 2^{2i} \, - \, 2^{-h} \big(2^{2i} 
			+ 2^{2h} \big) \, + \, 1 . 
	\end{align*}
Thus, it is enough to show that $2^{-h} ( 2^{2i} + 
2^{2h}) > 2^{i+1}$. 
But this follows since 
	\begin{equation*}
		2^{2i} \, + \, 2^{2h} 
		\ > \ 1 \geq \frac{2^{i+1}}{2^h} , 
	\end{equation*}
recalling that $h > i \geq 1$.
We conclude that the above mechanism satisfies 
\eqref{Eq:IC}.

To show that $\Rev(V'',W'') = \infty$, let $E_t 
\, = \, \{ (v_j,w_t) \mid v_j \geq 2^{2t} \}$. 
Notice that all types in the event $E_t$ choose the 
$t = k(v_j,w_t)$ menu entry.  
Since valuation and budget are independent random 
variables, we have that 
$\Prob \{E_t\} \, = \, \Prob \{V'' \geq 2^{2t}\} \, 
\Prob \{ W'' = 2^t\} \, = \, 2^{-t} (1/t)$.  
Now	
	\begin{align*}
		\Rev(V'',W'') 
		\ \geq \ \sum_{t = 1}^\infty p_t \, 
			\Prob \{E_t\}    
		\ = \ \sum_{t=1}^\infty \Big( 1 \, - \, 
			\frac{1}{2^t}\Big) \, 2^t \, 2^{-t} 
			\, \frac{1}{t}
		\ \geq \ \frac{1}{2} \sum_{t=1}^\infty \, \frac{1}{t} 
			\ = \ \infty .
	\end{align*}

Finally, to show that \autoref{Eq:6-fix-2} is satisfied, 
let $Z = \min(V'',W'')$.  
We want to establish an upper bound for $p \, \Prob \{Z 
\geq p\}$, for all $p \geq 0$. 

For $0 \leq p < 4$, this is immediate: $p \, \Prob \{ Z 
\geq p \} \leq 4 < \infty$. 
Now consider any $p \geq 4$.  
Choose an integer $n \geq 2$ such that $2^{n-1} < p \leq 2^n$. 
Then $p \, \Prob \{Z \geq p\} \leq 2^n \, \Prob \{ Z \geq 
2^{n-1}\}$.  
Since $V''$ and $V''$ are independent, we have that 
$\Prob\{Z \geq 2^{n-1}\} = \Prob\{V'' \geq 2^{n-1}\} 
\, \Prob\{W'' \geq 2^{n-1}\}$. 
Now, 
	\begin{equation*}
		\Prob\{W'' \geq 2^{n-1}\}
		\ = \ \sum_{i = n-1}^\infty 2^{-i} 
		\ = \ 2^{-(n-1)} \sum_{i=0}^\infty 2^{-i} 
		\ = \ 2^{-(n-2)} ,
	\end{equation*}
and further 
	\begin{equation*}
		\Prob\{V'' \geq 2^{n-1}\} 
		\ \leq \ \sum_{i = \lceil\frac{n-1}{2}\rceil}^\infty
			\frac{1}{i (i + 1)} 
		\ = \ \frac{1}{\lceil\frac{n-1}{2}\rceil} ,
	\end{equation*}
since the realization $V'' \geq 2^{n-1}$ is $2^{2i} \geq 
2^{n-1}$, that is $i \geq \lceil\frac{n-1}{2}\rceil$. 

We conclude that $\sup_p p\,\Prob\{Z \geq p\} \, \leq \, 
2^n \, 2^{-(n-2)} \lceil\frac{n-1}{2}\rceil^{-1} 
\, \leq \, 4 < \infty$, as desired. 
Thus, \autoref{Eq:6-fix-2} is satisfied. 
\end{proof}

For both the correlated random vector $(V',W')$ and 
the independent random vector $(V'',W'')$, we have seen 
that the optimal revenue is infinite.  
But, for any positive integer $m$, the optimal revenue 
in the subclass of mechanisms $\cali S_m$ is finite. 
Thus, there is no mechanism with finite menu size that 
guarantees a positive fraction of the optimal revenue. 
This concludes the proof of the proposition. 
\end{proof}


\subsection[Cash Bond]{Cash Bond Relaxations}
	\label{SS:cash-bond}

In this part of the paper we consider a relaxation of the
seller's revenue maximization problem.
Following \citet{Che:2000ux}, we study a setting where the 
seller can prevent the buyer from over-reporting her
budget at no cost.%
		\footnote{See Section 4 of their paper.}
In practice, ruling out budget over-reporting can be
accomplished by additional institutional arrangements.
For example, the seller can require the buyer to post a
cash bond prior to choosing a lottery, or by some
information disclosure rules (e.g., financial reports,
bank statements, etc).
In this setting, the incentive compatibility constraints
are required to hold only for under-reporting the private
budget.
As a result, higher budgeted types might be offered
better deals (discounts).

To formalize this possibility, let $\cali N_{cb}$ be the
class of ex-post budget feasible and ex-post individually
rational mechanisms that require a $(V,W)$--buyer
to post a cash bond to prevent her from over-reporting
her budget.
More explicitly, a mechanism $\nu = (x,s)$ belongs to
$\cali N_{cb}$ if it satisfies \eqref{Eq:BF}, \eqref{Eq:IR}
and the following constraint, which replaces
\eqref{Eq:IC}:
for all $(v,w)$ and all $(\tilde v, \tilde w)$ with
$\tilde w \leq w$, it must be that
	\begin{equation}
		v \, x(v,w) \, - \, s(v,w)
		 \ \geq \ v \, x(\tilde v,\tilde w) \, - \,
			s (\tilde v,\tilde w)   .
			\tag{CB} \label{Eq:CB}
	\end{equation}
In words, because the seller exacts a cash bond from
the buyer, the mechanism $\nu \in \cali N_{cb}$ only needs
to prevent a higher-budget type from mimicking a
lower-budget type.
It is immediate to realize that $\cali M \subseteq
\cali N_{cb}$.
Thus, for any family $\scr B$ of buyers, we have that
	\begin{equation*}
		\MVR (\cali N_{cb}; \scr B)
		\ \geq \ 1 ,
	\end{equation*}
since for any $(V,W)$-buyer in $\scr B$, it must be
$\Rev (V,W) \leq \Rev (V,W | \cali N_{cb})$.

From a computational point of view, the importance of
focusing on this relaxed problem comes from the fact that
it admits a natural linear programming representation
in discrete type spaces, which has been studied before
(see \citet{Bhattacharya:2010vy} for example).
In particular, \citet{Devanur:2017ik} characterize the
optimal mechanism in this relaxed setting and show that
it has a menu with exponentially-many non-trivial options
(in the number of possible budgets).

In principle, then, the optimal revenue obtained in the
superclass of mechanisms $\cali N_{cb}$ can be used to
approximate the optimal revenue in $\cali M$.
Of course, this provides a good approximation as long
as $\MVR (\cali N_{cb}; \scr B)$ has a reasonable bound.
\citet{Che:2000ux} show that if the random variables $V$
and $W$ are positively affiliated, then $\Rev(V,W)
= \Rev (V,W |\cali N_{cb})$.%
	\footnote{This result uses an additional \emph{declining
	marginal revenue} assumption --- see Assumptions
	1 and 2 in \citet{Che:2000ux}.}
Thus, when one focuses on the family of $(V,W)$--buyers
where the valuation and the budget are positively affiliated random
variables, there is no revenue gap between optimal
mechanisms in $\cali M$ and optimal approximation mechanisms
in $\cali N_{cb}$.

Unfortunately, positive affiliation is crucial for this good
approximation result to hold.  
\citet{Che:2000ux} provide an example showing that a higher
revenue can be extracted (compared to the optimal auction)
when the seller uses cash-bonds to prevent over-reporting
of private budgets.  
By adjusting this example, we show that the optimal revenue
in the class $\cali N_{cb}$ can be unboundedly large
compared to the optimal revenue in $\cali M$.
This shows that optimal approximation
mechanisms in the class $\cali N_{cb}$ cannot serve as
proxies to the revenue maximization problem in the class
$\cali M$.

\begin{proposition}  \label{P:negative-aff}
Let $\scr B_{na}$ be a family of $(V,W)$--buyers that
contains all distributions where $V$ and $W$ are negatively
affiliated.
Then one has
	\begin{equation*}
		\MVR (\cali N_{cb}; \scr B_{na})
		\ = \ + \, \infty .
	\end{equation*}
\end{proposition}

\begin{proof}
We consider a setting where the $(V,W)$--buyer has
a finite support given by
	\begin{equation*}
		\big\{ \big( v_i,w_i \big) \, = \,
			\big( 1/i, \, 1 \, + \, i/k \big)
			\,\colon\, i \, = \, 1,\ldots,k \big\} .
	\end{equation*}
Further, assume that the probability mass function $f$
for the $(V,W)$--buyer is
	\begin{equation*}
		f(v_i,w_i) \ \equiv \ f_i \ = \  1/k ,
			\qquad \text{for all $i = 1,\ldots,k$,}
	\end{equation*}
and zero otherwise.

Notice that for any buyer's type $(v_i,w_i)$, the valuation
is never higher than the budget, so the financial constraints
are never binding under truthful reporting.
There are many mechanisms in $\cali M$ that achieve the
optimal revenue in this case.
For instance, an optimal way to sell the item is by
posting a price of $1/i$ (for any $1 \leq i \leq k$).
Given said price, all types $(v_j,w_j), j = 1,\ldots,i$,
purchase the good.
Thus, $\Rev(V,W) = 1/k$.

Consider now a mechanism $\mu$ in the relaxed class
$\cali N_{cb}$ that charges $v_i$ for a guaranteed win of 
the object if the buyer posts a cash
bond equal to $w_i$, for all $i = 1,\ldots,k$.
Clearly, any buyer's type $(v_i,w_i)$ has no incentives
to under-report their budgets, since the budget is
increasing while the valuation is decreasing.
Therefore, when the seller can prevent over-reporting
with the use of cash bonds, the expected revenue is
	\begin{equation*}
		R (\mu ; V,W) \ = \
		\frac{1}{k} \Big( 1 \, + \, \frac{1}{2} \, + \,
		\ldots \, + \, \frac{1}{k} \Big)
		\ = \ \frac{H_k}{k} ,
	\end{equation*}
where $H_k$ is the $k$-th harmonic number.
Immediately, $\Rev (V,W | \cali N_{cb}) \geq H_k / k$.

Putting these two arguments together obtains the ratio
	\begin{equation*}
		\frac{\Rev (V,W | \cali N_{cb})}{\Rev (V,W)}
		\ \geq \ H_k , \qquad
			\text{for all $k = 1,2,\ldots$.}
	\end{equation*}

But as $k \to \infty$, we have that $H_k \to \infty$.
Thus, $\MVR (\cali N_{cb}; \scr B_{na})
\, = \, + \infty$.  
\end{proof}

What do we make of the result stated in
\autoref{P:negative-aff}?
In other words, should one take it as a positive or a negative
approximation result?
From a computational perspective, it is certainly not
positive.
Using the class of cash-bond mechanisms to approximate
the optimal revenue in $\cali M$, or equivalently ignoring
the incentive constraints that consider over-reporting,
can yield an exceedingly large overestimation of
the optimal revenue generated by a fully incentive
compatible mechanism.
At the same time, the fact that the maximal value
of relaxation in $\cali N_{cb}$ is potentially large may
provide a justification for the seller to look for
institutional fixes; i.e., lobby financial or other
regulatory authorities to make the use of cash bonds
legal.
However, caution should be made since if the seller fails 
to enforce no-over-reporting of budgets, then as 
demonstrated by the distribution in the proof of 
\autoref{P:negative-aff}, a slightly lower budgeted type 
might be better off taking out a small bridging loan 
(e.g., from friends and family), to end up  paying much 
less.

A related problem with the superclass of approximation 
mechanisms $\cali N_{cb}$ is that the important revenue 
monotonicity property is no longer generally valid.  
In fact, as we show next, non-monotonicities can be acute: 
it may be possible that some buyer dominates another 
in the first-order stochastic sense, yet expected revenue from 
the former is much higher than from the latter.   
We have the following result.

\begin{proposition}  \label{P:rev-non-mon-CB}
Let $\cali N_{cb}$ be the superclass of mechanism satisfying 
\eqref{Eq:BF}, \eqref{Eq:IR} and \eqref{Eq:CB}.  
The revenue non-monotonicity gap can be unboundedly 
large in $\cali N_{cb}$.  	
\end{proposition}

\begin{proof}
Consider the same $(V,W)$--buyer from 
\autoref{P:negative-aff}, for which we have shown that 
$\Rev (V,W | \cali N_{cb}) \geq H_k / k$, for all $k = 
1,2,\ldots$. 
Now consider the $(V',W')$--buyer with support 
	\begin{equation*}
		\big\{ \big( v_i',w_i' \big) \, = \,
			\big( 1/i, \, 2 \big)
			\,\colon\, i \, = \, 1,\ldots,k \big\} 		
	\end{equation*}
and probability mass function 
	\begin{equation*}
		f'(v_i',w_i') \ \equiv \ f'_i \ = \  1/k ,
			\qquad \text{for all $i = 1,\ldots,k$,}
	\end{equation*}
and zero otherwise.
Clearly $(V',W')$ first-order stochastically dominates 
$(V,W)$, since $v_i' = v_i$, $w_i' \geq w_i$ and 
$f_i' = f_i$.  

The optimal mechanism in the class $\cali N_{cb}$ for the 
$(V',W')$--buyer is a take-it-or-leave it offer for the 
entire good at a price of $1/k$.  
In fact, this fixed-price mechanism is optimal regardless 
of the seller's ability to prevent over-reporting of 
budgets.
Since the budget is identical across types and never binds, 
the seller cannot differentiate among buyers along the 
budget dimension. 
Therefore $\Rev(V',W' | \cali N_{cb}) = 1/k$.  

Putting these two observations together, we conclude 
	\begin{equation*}
		\frac{\Rev(V,W | \cali N_{cb})}
		{\Rev(V',W' | \cali N_{cb})} \ \geq \ H_k . 
	\end{equation*}
But one has $H_k \to \infty$ as $k \to \infty$, so that the 
revenue non-monotonicity gap can be unboundedly large.  
\end{proof}


\subsection[Strong IC]{Restricted Mechanisms under
Strong Incentive Constraints}   \label{SS:sic}

\citet{Daskalakis:2018zg} studied selling mechanisms when
the incentive compatibility constraints hold even if the
deviation produces a non-affordable outcome.
In other words, \citet{Daskalakis:2018zg} increase the number
of incentive constraints by considering deviations to
lotteries that are not ex-post affordable for a buyer given
her budget.
Formally, they consider the following incentive constraint
on the mechanism $\mu = (x,s)$: for all $(v,w) \in \R_+^2$
and all $(\tilde v, \tilde w) \in \R_+^2$, it must be that
	\begin{equation}
		v \, x(v,w) \, - \, s(v,w)
		 \ \geq \ v \, x(\tilde v,\tilde w) \, - \,
		s (\tilde v,\tilde w) .  \tag{SIC} \label{Eq:SIC}
	\end{equation}
Comparing it with \eqref{Eq:IC}, this constraint omits the
requirement that the deviation of the buyer's type
$(v,w)$ to purchasing a lottery for the type
$(\tilde v, \tilde w)$ must be within $(v,w)$'s budget.

Let $\cali N_{sic} \subseteq \cali M$ denote the class of
mechanisms that satisfy \eqref{Eq:BF}, \eqref{Eq:IR} and
\eqref{Eq:SIC}.
Of course, the imposition of these `extra' incentive
constraints means that, for any $(V,W)$--buyer, the optimal
revenue in the class $\cali N_{sic}$ will be weakly lower
than the optimal revenue in the class $\cali M$.
The advantage of imposing these constraints is that the
resulting setting has a natural linear programming
formulation.
So, in principle, it can be used to approximate
the revenue generated by the optimal mechanisms in
$\cali M$.

Our two next results on revenue gaps in the presence of
budget constraints show that, sometimes, the
expected revenue of the optimal approximation mechanism
in the subclass $\cali N_{sic}$ can be arbitrarily
small compared to the expected revenue of the optimal
mechanism.
Thus, despite its computational
advantages, the subclass $\cali N_{sic}$ of mechanisms
with the strong incentive constraint cannot serve as a
good proxy for the class $\cali M$ of incentive compatible,
ex-post budget feasible and ex-post individually
rational mechanisms.
Below we argue that this holds at least in two important cases:
when one considers the family of $(V,W)$--buyers with
unbounded support (\autoref{P:menu-SIC}), or
the family of $(V,W)$--buyers with support
in the unit square (\autoref{P:menu-SIC-square}).

\begin{proposition}  \label{P:menu-SIC}
Let $\scr B_{u}$ be the family of $(V,W)$--buyers that
contains all distributions with unbounded support.
Then one has that
	\begin{equation*}
		\GFOR (\cali N_{sic}; \scr B_{u}) \ = \ 0 .
	\end{equation*}
\end{proposition}

\begin{proof}
Consider a setting where the buyer has only two possible types.
Thus, let $(V,W)$ be a random vector with
support $\{(v_1,w_1),(v_2,w_2)\}$ given by
	\begin{equation*}
		(v_1,w_1) \ = \ (B + \epsilon, B)
			\quad \text{and} \quad
		(v_2,w_2) \ = \ (B + 2\epsilon, 1) . 		
	\end{equation*}
Here we assume that $B > 1$ and $\epsilon > 0$.
Note that both types are financially constrained, i.e.,
in each case the valuation is higher than the budget.
The first type, which occurs with probability $f_1 \equiv
f(v_1,w_1) = 1 - 1/B$, is however less financially
constrained than the second type (which occurs with
complementary probability $f_2 = 1 - f_1)$.

We now solve the linear program that obtains
$\Rev (V,W | \cali N_{sic})$.
As usual, we start with a direct mechanism $(x_i,s_i)$,
where $x_i \in [0,1]$ is the probability of allocating
the item to type $i = 1,2$ and $s_i$ is the expected
transfer collected by the seller.
The linear program for the seller can be expressed by
	\begin{equation*}
		\max \, \big(1 \, - \, \tfrac{1}{B} \big)\, s_1
			 \ + \ \tfrac{1}{B} \, s_2 \qquad
	\end{equation*}
subject to the following constraints
	\begin{align}
		0 & \ \leq \ x_1 \ \leq \ 1 ,
			&
		0 & \ \leq \ x_2 \ \leq \ 1 	,		
			\label{Eq:strong-x} \\[.75ex]
		0 & \ \leq \ s_1  ,
			&
		0 & \ \leq \ s_2  ,		
			\label{Eq:strong-t} \\[.75ex]
		s_1 & \ \leq \ (B \, + \, \epsilon) \, x_1 ,
			&
		s_2 & \ \leq \ (B \, + \, 2\epsilon) \, x_2 ,
			\quad  \label{Eq:strong-IR} 	\\[.75ex]			
		s_1 & \ \leq \ B \, x_1 ,
			&
		s_2 & \ \leq \ x_2 ,   \label{Eq:strong-BF}
	\end{align}
and finally \vspace{-2ex}
	\begin{align}
		(B \, + \, \epsilon) x_1 \, - \, s_1
			& \ \geq \ (B \, + \, \epsilon)x_2 \, - \, s_2 ,
			\label{Eq:strong-IC-d}  \\[.75ex]
		(B \, + \, 2 \epsilon) x_2 \, - \, s_2
			& \ \geq \ (B \, + \, 2 \epsilon) x_1 \, - \, s_1 .
			\label{Eq:strong-IC-u}
	\end{align}

The first four constraints, expressed in Eq.~\eqref{Eq:strong-x}
and Eq.~\eqref{Eq:strong-t}, are just the standard restrictions
to the allocation probabilities and expected transfers
(recall the seller does not subsidize the buyer's consumption).
The constraints in Eq.~\eqref{Eq:strong-IR} are the
participation constraints, and the ones in
Eq.~\eqref{Eq:strong-BF} are the ex-post budget feasibility
constraints.
It's more interesting to consider the last two
incentive compatibility constraints.
\autoref{Eq:strong-IC-d} prevents type $(v_1,w_1)$
from selecting the lottery designed for type $(v_2,w_2)$;
likewise \autoref{Eq:strong-IC-u} prevents $(v_2,w_2)$
from choosing the lottery designed for $(v_1,w_1)$.
In both constraints, no budget affordability is taken
into account.

We first observe that the mechanism $\{(x_1,s_1) = (1,B),
(x_2,s_2) = (0,0)\}$, i.e., selling the good for a posted
price of $B$, is not a feasible solution to the above
linear program.
To see this, notice that this mechanism violates
the incentive constraint in \autoref{Eq:strong-IC-u}.
In reality, of course, type $(v_2,w_2)$ cannot afford
to pay the price of $B$ for the object, since her budget
is just $1 < B$, but the incentive constraint does not take
affordability into account.

Next, we claim that the optimal solution to the above
linear program is $(x_1^*,s_1^*) = (x_2^*,s_2^*) =
(1,1)$.
Indeed, combining \autoref{Eq:strong-IC-d} and
\autoref{Eq:strong-IC-u} we observe that any feasible
solution will have $x_2 \geq x_1$.
Now, from the incentive constraint in \autoref{Eq:strong-IC-d}
we obtain that
	\begin{align*}
		s_1
		\ \leq \ (B \, + \, \epsilon) (x_1 \, - \,
			x_2) \, + \, s_2
		\ \leq \ s_2 \ \leq \ x_2 \ \leq \ 1 ,
	\end{align*}
where the second inequality follows from our previous observation,
and the two remaining inequalities come from
\autoref{Eq:strong-BF} and \autoref{Eq:strong-x}.
In particular, in any feasible solution $s_1 \leq 1$ and
$s_2 \leq 1$.

Finally, note that setting $s_1^* = s_2^* = 1$ and
$x_1^* = x_2^* = 1$ constitutes a feasible solution in the
class of mechanisms $\cali N_{sic}$.
Thus, this is the optimal solution for the linear program.
Clearly, the optimal mechanism can be implemented by a posted
price equal to one, which is affordable to both types.
We conclude that $\Rev (V,W | \cali N_{sic}) = 1$.
However, the optimal mechanism under the `correct'
incentive constraints can achieve an expected revenue
of at least $B (1 - 1/B) = B - 1$, for instance by using a
fixed price of $B$.
This excludes type $(v_2,w_2)$ who cannot afford to
buy the good at this price.
Thus, $\Rev(V,W) \geq B - 1$.

We conclude that, for the $(V,W)$--buyer,
	\begin{equation*}
		\frac{\Rev (V,W | \cali N_{sic})}{\Rev (V,W)}
		\ \leq \ \frac{1}{B \, - \, 1} .
	\end{equation*}
Since the family of buyers $\scr B_u$ contains random
vectors with unbounded support, taking $B \to \infty$ shows
that $\GFOR (\cali N_{sic}; \scr B_u) = 0$, as claimed.  
\end{proof}

\smallskip

\begin{proposition}  \label{P:menu-SIC-square}
Let $\scr B_{0}$ be a family of $(V,W)$--buyers that
contains all distributions whose supports are subsets
of $[0,1] \times [0,1]$.
One has that
	\begin{equation*}
		\GFOR (\cali N_{sic}; \scr B_{0}) \ = \ 0 .
	\end{equation*}
\end{proposition}

\begin{proof}
Fix $\epsilon > 0$ and let $B > 1$.
Consider now a random vector $(\hat V, \hat W)$ in $\scr B_0$
with support
$\{(\hat v_1,\hat w_1), (\hat v_2,\hat w_2)\}$ given by
	\begin{equation*}
		(\hat v_1, \hat w_1)
		\ = \ \left( \tfrac{B \, + \, \epsilon}
			{B \, + \, 2\epsilon}, \tfrac{B}{B \, + \,
			2 \epsilon} \right)
		\qquad \text{and} \qquad
		(\hat v_2, \hat w_2) \ = \ \left( 1,
			\tfrac{1}{B \, + \, 	2 \epsilon} \right) .
	\end{equation*}
Assume further that $f(\hat v_1,\hat w_1) = 1 - 1/B$ and
$f(\hat v_2,\hat w_2) = 1/B$.
Note that this $(\hat V, \hat W)$--buyer in the family
$\scr B_0$ has been obtained from the $(V,W)$--buyer
in the proof of \autoref{P:menu-SIC} by scaling it down
by the factor $1/(B + 2\epsilon)$.

A simple adaptation of the proof of \autoref{P:menu-SIC}
shows that the optimal approximation mechanism in
$\cali N_{sic}$ here specifies allocation probabilities
$x_1^* = x_2^* = 1$ and expected payments $s_1^* = s_2^*
= 1/(B + 2\epsilon)$.
Therefore, $\Rev (\hat V, \hat W | \cali N_{sic})
= 1/(B + 2\epsilon)$.
Similarly, charging a fixed price of $B/(B+2\epsilon)$
for the object raises $(B-1)/(B + 2\epsilon)$ for the seller,
and thus $\Rev(\hat V, \hat W) \geq (B-1)/(B+2\epsilon)$.
It follows that, as before, one has
	\begin{equation*}
		\frac{\Rev (\hat V,\hat W | \cali N_{sic})}
			{\Rev (\hat V,\hat W)}
		\ \leq \ \frac{1}{B \, - \, 1} .
	\end{equation*}
Letting $B \to \infty$ obtains the result.  
\end{proof}

These bad approximation results in terms of the revenue gap
between computationally feasible mechanisms in $\cali N_{sic}$
and optimal mechanisms in $\cali M$ mimic the bad
approximation results for simple mechanisms previously
obtained.
Furthermore, the subclass $\cali N_{sic}$ shares 
the revenue non-monotonicity drawback that we observed in 
the superclass $\cali N_{cb}$. 
This is an unexpected connection between a class of mechanisms 
that imposes \emph{more} constraints than the standard 
\eqref{Eq:IC}, \eqref{Eq:IR} and \eqref{Eq:BF} (namely 
$\cali N_{sic}$) and a class of mechanisms that imposes 
less constraints (namely $\cali N_{cb}$).

\begin{proposition}  \label{P:rev-non-mon-SIC}
Let $\cali N_{sic}$ be the subclass of mechanism satisfying 
\eqref{Eq:BF}, \eqref{Eq:IR} and \eqref{Eq:SIC}.  
The revenue non-monotonicity gap can be unboundedly 
large in $\cali N_{sic}$.  	
\end{proposition}

\begin{proof}
Fix some $H > 1$. 
Consider a $(V,W)$--buyer whose realized type is $(H,1)$ with 
probability $1/H$ and $(H,H)$ with the complementary 
probability $1 - 1/H$. 
We shall argue that $\Rev(V,W | \cali N_{sic}) \geq H-1$.  
To see this, we formulate the seller's problem 
with the linear program below and show that there exists 
a feasible solution with an objective value of $H-1$. 
The linear program is 
	\begin{equation*}
		\max \, \tfrac{1}{H} \, s_1
			 \ + \ \big(1 - \tfrac{1}{H}\big) \, s_2 \qquad
	\end{equation*}
subject to the following constraints
	\begin{align*}
		0 & \ \leq \ x_1 \ \leq \ 1 ,
			&
		0 & \ \leq \ x_2 \ \leq \ 1 	,		
			 \\[.75ex]
		0 & \ \leq \ s_1  ,
			&
		0 & \ \leq \ s_2  ,		
			 \\[.75ex]
		s_1 & \ \leq \ H \, x_1 ,
			&
		s_2 & \ \leq \ H \, x_2 ,
			\quad   	\\[.75ex]			
		s_1 & \ \leq \ 1 \, x_1 ,
			&
		s_2 & \ \leq \ H \, x_2 ,   
	\end{align*}
and finally the strong IC constraints, which can be written as 
	\begin{equation*}
		H \, x_1 \, - \, s_1 \ = \ 
		H \, x_2 \, - \, s_2 . 
	\end{equation*}
Clearly, $(x_1,s_1) = (0,0)$ and $(x_2,s_2) = (1,H)$ 
is a feasible solution at which the objective function 
attains a value of $H - 1$. 

Now let $\epsilon > 0$ and consider a $(V',W')$--buyer 
whose realized type is $(H+\epsilon,1)$ with probability 
$1/H$ and $(H,H)$ with the complementary probability.  
Note that $(V',W') \geq (V,W)$.  
The linear program for the seller can be expressed by
	\begin{equation*}
		\max \, \tfrac{1}{H} \, s_1
			 \ + \ \big(1 - \tfrac{1}{H}\big) \, s_2 \qquad
	\end{equation*}
subject to the following constraints
	\begin{align*}
		0 & \ \leq \ x_1 \ \leq \ 1 ,
			&
		0 & \ \leq \ x_2 \ \leq \ 1 	,		
			 \\[.75ex]
		0 & \ \leq \ s_1  ,
			&
		0 & \ \leq \ s_2  ,		
			 \\[.75ex]
		s_1 & \ \leq \ (H + \epsilon) \, x_1 ,
			&
		s_2 & \ \leq \ H \, x_2 ,
			\quad   	\\[.75ex]			
		s_1 & \ \leq \ 1 \, x_1 ,
			&
		s_2 & \ \leq \ H \, x_2 ,   
	\end{align*}
and finally the strong IC constraints
	\begin{align}
		(H + \epsilon)x_2 \, - \, s_2 
			\, - \, (H+\epsilon)x_1 \, + \, s_1 
			& \ \leq \ 0 ,  \label{Eq:non-mon-SIC-1}  \\[.75ex] 
		H\,x_1 \, - \, s_1 \, - \, H\,x_2 
			\, + \, s_2 
			& \ \leq \ 0 . \label{Eq:non-mon-SIC-2}
	\end{align}

To simplify calculations, let $\alpha = (H-1)/\epsilon$ and 
$\beta = \alpha + 1 - 1/H$. 
Clearly $\alpha,\beta > 0$.  
Notice that $\alpha(H+\epsilon) - \beta H = 0$. 
From the feasibility constraints, we have $1 \geq x_1$. 
Add this and the budget feasibility constraint $0 \geq 
s_1 - 1 x_1$ to obtain $1 \geq s_1$.  
Multiply $\alpha$ to \autoref{Eq:non-mon-SIC-1} and 
$\beta$ to \autoref{Eq:non-mon-SIC-2}, and combine them 
with the previous expression to obtain 
	\begin{align*}
		1 
		& \ \geq \ s_1 \, + \, \alpha [(H+\epsilon)x_2 
			\, - \, s_2 \, - \, (H+\epsilon)x_1 
			\, + \, s_1 ] \ + \ \beta [ Hx_1 - s_1 
			\, - \, Hx_2 \, + \, s_2 ] \\
		& \ = \ s_1 \, + \, (\alpha - \beta)s_1 
			\, + \, (\beta - \alpha)s_2 \, + \, 
			[ \beta H - \alpha(H+\epsilon)](x_1 - x_2) \\
		& \ = \ \tfrac{1}{H}s_1 \, + \, \big(1 - 
			\tfrac{1}{H}\big) s_2 . 
	\end{align*}
This shows that the objective function of the 
seller's linear program for the $(V',W')$--buyer has a 
value of at most $1$.  
Consider the feasible solution $(x_1^*,s_1^*) = (1,1)$ 
and $(x_2^*,s_2^*) = (1,1)$.  
Clearly, $\frac{1}{H}s_1^* + \big(1 - \frac{1}{H}\big) 
s_2^* = 1$, and thus $\Rev(V',W' | \cali N_{sic}) = 1$.  

These arguments show that, for any given $H > 1$, 
	\begin{equation*}
		\frac{\Rev(V,W) | \cali N_{sic})}
		{\Rev(V',W') | \cali N_{sic})} 
		\ \geq \ \frac{H - 1}{1} . 
	\end{equation*}
Letting $H \to \infty$ obtains the desired conclusion. 
\end{proof}

Why do (some) approximation mechanisms violate 
revenue monotonicity? 
The environment we construct in the proof of 
\autoref{P:rev-non-mon-SIC} considers a random vector 
$(V',W')$ whose support is not a lattice. 
Because of the extra constraints embedded in \eqref{Eq:SIC}, 
we cannot apply the \emph{Extension Lemma} to expand the 
domain of these mechanisms to $\R^2_+$. 
Thus, even though the mechanisms we consider there are 
in $\cali M$ (because $\cali N_{sic}$ is a subclass of 
$\cali M$), we cannot apply \autoref{P:rev-mon} to 
obtain revenue monotonicity.


\section{Final Remarks}

In this paper, we explore the advantages and limitations of using
approximation mechanisms in settings with private valuations
and private budgets.
In most of the cases we consider, there is a significant
revenue gap between optimal and approximately optimal mechanisms.
Thus, despite the advantages (conceptual, computational,
and practical) of using a restricted or relaxed class
of mechanisms, one must consider the large potential 
revenue losses.

A central message of our results is that the limitations we 
identify are driven by the sensitivity of approximation 
guarantees to seemingly mild features of the type space. 
When valuations and budgets are jointly bounded above and 
bounded away from zero, the combination of revenue monotonicity 
and a discretization of the type space yields 
near-optimal performance with polylogarithmic menu size. 
Outside this narrow setting, however, the simplicity--optimality 
tradeoff deteriorates sharply: once unbounded support is 
permitted, or even when attention is restricted to distributions 
supported on the unit square, no mechanism with a constant 
in size (or asymptotically sublinear) menu can guarantee a positive 
fraction of optimal revenue. 
In particular, \autoref{P:sublinear} shows that this failure 
persists even when menu size is allowed to grow sublinearly 
with support size, while \autoref{P:infinite-supp}
shows that for certain infinite-support distributions no 
finite-menu mechanism can guarantee any positive fraction 
of the optimal revenue, even under independence.
In this sense, private budgets turn the single-item problem 
into a genuinely two-dimensional screening environment 
in which robustness to distributional details is costly, 
and ``simple'' mechanisms are generically far from 
approximately optimal.
Our analysis also highlights that computational or institutional 
relaxations should be interpreted with care when used 
as substitutes for the fully incentive compatible benchmark. 
The cash-bond relaxation can generate unbounded gaps relative 
to what is implementable without additional enforcement, 
and it can break revenue monotonicity.  
Likewise, strengthening incentive constraints to 
obtain linear-programming formulations can lead to severe 
revenue losses and violations of revenue monotonicity, even 
in small discrete environments.

These observations point to several natural directions for 
future work: identifying intermediate notions of verifiability 
(e.g., partial auditing, probabilistic verification, or 
limited penalties for misreporting) under which one can 
recover meaningful approximation guarantees; developing 
distributional conditions beyond bounded support that still 
imply bounded revenue gaps; and extending the menu-complexity 
analysis to richer environments with multiple bidders or 
multiple items, where budgets interact with competition and 
allocation externalities in ways that may further amplify 
---or, perchance, mitigate--- the gaps documented here.


\bibliography{revenue-gaps}


\renewcommand\thesection{\Alph{section}}

\setcounter{section}{0}

\vspace{3ex}  \small 

\begin{center}
	\Large \textbf{APPENDIX}
\end{center}


\hypertarget{App-A}{\section{Extension Lemma}}

In this appendix we show that the \emph{Extension Lemma} 
of \citet{Hart:2015ob} applies to a setting with private 
valuations and private budgets.
Let $D \subseteq \R^2_+$ be the type domain,
i.e., the support of the random vector $(V,W)$.
Let $\mu = (x,s)$ be a direct mechanism defined on $D$,
and let
	\begin{equation*}
		\Menu (\mu)\ := \ \big\{ (x(v,w),s(v,w)) \,\colon\,
			(v,w) \in D \big\} \ \subseteq \ 
            [0,1] \times \R_+
	\end{equation*}
represent its menu.
Notice we are restricting the payments to the seller
to be non-negative.
Suppose $\mu = (x,s)$ is ex-post budget feasible \eqref{Eq:BF}
and incentive compatible \eqref{Eq:IC}.
Denote the buyer's indirect utility (expected payoff)
under $\mu$ by $b(v,w) = x(v,w)\,v - s(v,w)$, for all
types $(v,w) \in D$.

For each $\tilde w \in \R_+$, define
	\begin{equation}  \label{Eq:M(w)}
		M(\tilde w)
		\ := \ \big\{ (q,p) \, \in \, \Menu(\mu)
		\,\colon\, p \, \leq \, \tilde w \, q \big\} .
	\end{equation}
In words, $M(\tilde w)$ is the subset of $\Menu(\mu)$
restricted to entries that are ex-post affordable
given a budget $\tilde w \geq 0$, even when $\tilde w$
is not part of any realized type in $D$.
It is clear that the mapping $\tilde w \mapsto M(\tilde w)$
is monotone in the set-inclusion sense: $0 \leq w' \leq w''$
implies $M(w') \subseteq M(w'')$.
Readily,
	\begin{equation*}
		\bigcup_{\tilde w \geq 0} M(\tilde w)
		\ = \ \Menu(\mu) .
	\end{equation*}

We can use (some of) these submenus to express the expected
payoff function for the buyer.
For all $(v,w) \in D$, write
	\begin{align*}
		b(v,w)
		& \ = \ \max_{(v',w') \in D} \big\{ v\,x(v',w')
			\, - \, s(v',w') \,\colon\, s(v',w') \,
			\leq \, w \, x(v',w') \big\} \\
		& \ = \ \max \big\{ v\,q \, - \, p \,\colon\,
			(q,p) \,\in\, M(w) \big\} .
	\end{align*}

\begin{lemma}  \label{L:extension}
Let $\mu = (x,s)$ be an ex-post budget feasible and
incentive compatible mechanism defined on $D \subseteq
\R^2_+$ with $\Menu(\mu) = (x,s)(D) \subseteq [0,1]
\times \R_+$.
Then $\mu$ can be extended to an ex-post budget feasible
and incentive compatible mechanism $\overline\mu =
(\overline x, \overline s)$ defined on $\R_+^2$ with
$\Menu(\overline \mu) = (\overline x, \overline s)(\R_+^2)$
satisfying $\Menu(\mu) \subseteq \Menu(\overline\mu)
\subseteq \cl \Menu(\mu)$.
\end{lemma}

\begin{proof}
Let $\mu = (x,s)$ be an ex-post budget feasible and
incentive compatible mechanism with $\Menu(\mu)$.
Observe that, for each $w \geq 0$, the projection of
the set $M(w)$ defined in \autoref{Eq:M(w)} onto the
unit interval is clearly bounded.
Since $\mu$ specifies non-negative (nonlinear) prices,
and prices in the submenu $M(w)$ are bounded above
by $w$, it follows that $M(w)$ is a bounded set.

For each $w \geq 0$, define the function $\overline
b(\cdot,w)$ on $\R_+$ by
	\begin{equation}  \label{Eq:payoff-extension}
		\overline b(v,w)
		\ := \ \sup \big\{ v\,q \, - \, p
			\,\colon\, (q,p) \in M(w) \big\} .
	\end{equation}
Since, as we just argued, $M(w)$ is bounded,
for each $v \geq 0$ the supremum $\overline b(v,w)$
in \autoref{Eq:payoff-extension} is attained on the closure
of $M(w)$, say at $(\overline x(v,w),\overline s(v,w))
\in \cl M(w)$.
Thus, for every $(v,w) \in \R_+^2$ we can write
	\begin{align*}
		\overline b(v,w)
		& \ = \ v\,\overline x(v,w) \, - \, \overline s(v,w) \\
		& \ = \ \max \big\{ v\,q \, - \, p
			\,\colon\, (q,p) \in \cl M(w) \big\}  \\
		& \ = \ \max \big\{ v \, \overline x(v',w') \, - \,
			\overline s(v',w') \,\colon\, 0 \leq v' , \,
				0 \ \leq w' \leq w \big\}.
	\end{align*}
This shows that the mechanism $\overline \mu = (\overline x,
\overline s)$ is incentive compatible on $\R_+^2$.
By construction (see \autoref{Eq:payoff-extension}), the
direct mechanism $\overline\mu$ is also ex-post feasible
for the buyer.
Readily, we have that $\overline b(v,w) = b(v,w)$ for
every $(v,w) \in D$, since $\mu = (x,s)$ is incentive
compatible.
Thus, $\overline\mu$ extends $\mu$ to $\R_+^2$.

Finally, notice that
	\begin{equation*}
		\bigcup_{w \geq 0} M(w)
		\ \subseteq \ \bigcup_{w \geq 0} \cl M(w)
		\ \subseteq \ \text{cl}\left(\bigcup_{w \geq 0} 
        M(w)\right),
	\end{equation*}
which shows that $\Menu(\mu) \subseteq \Menu(\overline\mu)
\subseteq \cl \Menu(\mu)$, as desired.  
\end{proof}

\begin{corollary}
If the ex-post feasible and incentive compatible mechanism
$\mu = (x,s)$ defined on $D \subseteq \R_+^2$ is in
addition ex-post individually rational, then so is its
extension $\overline\mu = (\overline x, \overline s)$
constructed in \autoref{L:extension}.
\end{corollary}

\begin{proof}
Immediate from the fact that $\Menu(\overline\mu)$ is
a subset of the closure of $\Menu(\mu)$.   
\end{proof}


\hypertarget{App-B}{\section{On the Optimal Mechanisms for 
\autoref{Ex:public-v}}}

In this appendix, we show that every optimal ex-post 
mechanism for \autoref{Ex:public-v} must have an infinite 
menu. 
We consider a single bidder with \textit{publicly-known} value 
and private budget. 
Let the buyer's valuation be public and equal to $\hat v > 1$. 
Let the buyer's budget $W$ be private and distributed 
continuously on $[0,1]$.  
We assume a uniform distribution for the explicit revenue 
calculation below; the structural optimality argument does 
not require uniformity.

\begin{lemma}  \label{L:optimality-CG-1}
The optimal mechanism for \autoref{Ex:public-v} requires 
infinitely many lotteries. 
\end{lemma}

\begin{proof}
Suppose, to obtain a contradiction, that some optimal 
feasible mechanism is implemented by a finite nontrivial menu
	\[
		M \ = \ \{(q_1,p_1),\ldots,(q_k,p_k)\},
	\]
where each $q_i\in(0,1]$ is an allocation probability and 
each $p_i>0$ is an expected payment.

For any nontrivial lottery $(q_i,p_i)\in M$, define its 
actual price by $t(q_i,p_i) = p_i/q_i$.  
A type with budget $w \in [0,1]$ can afford $(q_i,p_i)$ if and 
only if $p_i\le wq_i$, or equivalently, if and only if 
$t(q_i,p_i) \le w$.
If affordable, the utility from $(q_i,p_i)$ is
	\begin{equation*}
		U(q_i,p_i)
		\ = \ \hat v \, q_i \, - \, p_i
		\ = \ q_i \bigl(\hat v \, - \, t(q_i,p_i)\bigr).		
	\end{equation*}
We call a menu entry $(q_i,p_i)\in M$ \emph{relevant} if there 
exists some budget type $w\in[0,1]$ for which $(q_i,p_i)$ is 
chosen among all entries in $M$ that are affordable to type $w$.

We first observe that entries with the same actual price 
are redundant. 
Suppose two nontrivial menu entries $(q_i,p_i)$ and $(q_j,p_j)$ 
have the same actual price: $p_i/q_i = p_j/q_j = t$.
Then they are affordable for exactly the same set of types, 
namely those with $w \geq t$.  
Since $t \leq 1 < \hat v$, we have $\hat v - t > 0$, so 
the utility is strictly increasing in $q$ for fixed $t$. 
Without loss of generality, assume that $q_i \ge q_j$. 
If $q_i=q_j$, the entries are identical; if $q_i>q_j$, then
	\begin{equation*}
		U(q_i,p_i) 
		\ = \ q_i(\hat v - t) \ > \ q_j(\hat v - t) 
		\ = \ U(q_j,p_j),		
	\end{equation*}
meaning $(q_j,p_j)$ is never chosen. 
Therefore, after deleting duplicate and irrelevant entries, we
may assume all remaining entries have pairwise distinct 
actual prices.

We re-label the remaining relevant entries as $(q_1,p_1), 
\ldots,(q_k,p_k)$ such that $0 < t_1 < t_2 < \cdots < t_k 
\leq 1$, where $t_i = p_i/q_i$ for all $i = 1,\ldots,k$.
The inequality $t_i \le 1$ holds because each relevant entry 
must be affordable for at least one type in $[0,1]$.
Define the utility levels by $U_i = \hat v q_i - p_i$,
and $U_0 = 0$ (the outside option). 
We claim that
	\begin{equation*}
		0 \ = \ U_0 \ \le \ U_1 \ \le \ U_2 \ \le \ 
		\cdots \ \le \ U_k.		
	\end{equation*}
Indeed, fix $i<j$. 
Since $t_i<t_j$, any type that can afford entry $j$ can also 
afford entry $i$. 
If $U_i > U_j$, every type that could afford $j$ would strictly 
prefer $i$, contradicting the relevance of $j$.  
Hence $U_i \le U_j$. 
Furthermore, since $q_1 > 0$ and $t_1 \le 1 <\hat v$ we have
	\begin{equation*}
		U_1 \ = \ \hat v q_1 \, - \, p_1 
		\ = \ q_1(\hat v \, - \, t_1) \ > \ 0.		
	\end{equation*}
If $w<t_1$, no nontrivial menu entry is affordable and the 
buyer can only choose the outside option.
Thus the utility of every $w \in (0,t_1)$ is zero and the 
mechanism yields zero revenue on the interval $(0,t_1)$.

We now construct an improved mechanism.

Choose any continuous strictly increasing function
$\tilde u \colon (0,t_1) \to(0,U_1)$.  
For each $w\in(0,t_1)$, define a lottery $L_w = (x(w),s(w))$ by 
	\begin{equation*}
		x(w) \ = \ \frac{\tilde u(w)}{\hat v-w} 
		\qquad \text{and} \qquad 
		s(w) \ = \ \frac{w\,\tilde u(w)}{\hat v-w}.
	\end{equation*}
Construct a new mechanism $\tilde M$ as follows:
\begin{itemize}
	\item for types $w\ge t_1$, keep the original menu $M$ 
		unchanged;
	\item for types $w\in(0,t_1)$, assign the lottery $L_w$.
\end{itemize}

Since $0 < \tilde u(w) < U_1 = q_1(\hat v - t_1) \le \hat v 
- t_1 < \hat v-w$, we have $0<x(w)<1$. 
Also, notice that $s(w)/x(w) = w$, so $L_w$ is affordable 
for type $w$, and its utility is $\hat v x(w) - s(w) = \tilde 
u(w) > 0$.
Thus \eqref{Eq:IR} and \eqref{Eq:BF} hold for types $w\in(0,t_1)$. 
It remains to verify \eqref{Eq:IC}.

First, consider deviations by types in $(0,t_1)$.  
Fix $w \in (0,t_1)$. 
If $w' > w$, then $L_{w'}$ has actual price $w' > w$ and is 
not affordable to type $w$. 
If $w'\le w$, then $\tilde u(w') < \tilde u(w)$ because
$\tilde u(\cdot)$ is strictly increasing. 
Thus among all new lotteries affordable to type $w$, the 
truthful lottery $L_w$ is strictly preferred. 
No old nontrivial menu entry is affordable to type $w$, 
because every old nontrivial entry has actual price at 
least $t_1>w$.

Next, consider deviations by types in $[t_1,1]$. 
Every such type can afford the original first menu
entry, which yields utility $U_1$. 
Every new lottery $L_w$ with $w < t_1$ yields utility 
$\tilde u(w) < U_1$. 
Therefore no type in $[t_1,1]$ wants to deviate to any 
new lottery. 

We have argued that $\tilde M$ is feasible and incentive 
compatible.  
It remains to show that $\tilde M$ yields strictly higher 
revenue than $M$. 
Under $M$, revenue on $(0,t_1)$ was zero. 
Under $\tilde M$, every type $w \in (0,t_1)$ generates 
strictly positive expected payment $s(w) = w\,\tilde u(w)/
(\hat v - w) > 0$. 
Since $(0,t_1)$ has positive measure under the continuous 
distribution, the expected revenue strictly increases.
This contradicts the optimality of the original finite 
menu $M$. 
Therefore no optimal feasible mechanism can be implemented 
by a finite nontrivial menu. 
\end{proof}


\end{document}

%% file: revenue-gaps-macros.tex
\usepackage{xcolor}
\usepackage[colorlinks,linkcolor=Blue,citecolor=Blue,
	urlcolor=Blue]{hyperref}
\usepackage{amssymb,latexsym,bm,mathrsfs,dsfont}
\usepackage{amsmath,amsfonts,amsthm}
\usepackage[square,numbers]{natbib}  
\usepackage{textcomp}

\usepackage{epigraph}
\usepackage{graphics,graphicx}
\graphicspath{{graphics/}}
\usepackage{amsbsy}
\usepackage{setspace}
\usepackage{cleveref}
\usepackage{enumitem}
\usepackage{hyperref}

\usepackage[linecolor=red,backgroundcolor=red!25, 
	bordercolor=red,prependcaption,textsize=tiny]{todonotes}

\usepackage{afterpage}

\usepackage{abstract}
	 
\usepackage[en-US]{datetime2}		
	\DTMlangsetup{showdayofmonth=false} 

\usepackage{stmaryrd}   

\usepackage[many]{tcolorbox}

\usepackage{lipsum}

\usepackage{subfig}

\usepackage[margin=1.0in]{geometry}
\definecolor{Blue}{rgb}{0.1,0.2,0.5}
\definecolor{Magenta}{cmyk}{0.1,0.9,0.0,0.3}
\definecolor{cites}{HTML}{324b13}
\definecolor{links}{HTML}{1a663b}

\definecolor{ocre}{RGB}{243,102,25} 
\definecolor{lightgray}{RGB}{229,229,229} 
\definecolor{mediumgray}{RGB}{120,120,120}

\definecolor{darkblue}{rgb}{0.0, 0.18, 0.39}
\definecolor{burgundy}{rgb}{0.5, 0.0, 0.13}
\definecolor{armygreen}{rgb}{0.29, 0.33, 0.13}

\usepackage{pgfplots}
\pgfplotsset{width = 7cm, compat=1.5}
\usepackage{tikz}
\usetikzlibrary{calc}
\usepackage{color}
\usepackage{array}
\usepackage{multirow}
\usepgflibrary{arrows}
\pgfplotsset{compat=1.6}
\pgfplotsset{soldot/.style={color=black,only marks,mark=*}} 
\pgfplotsset{holdot/.style={color=black,fill=white,only marks,mark=*}}
\pgfplotsset{every axis/.append style={
                axis x line=center,
                axis y line=center,
                axis line style={->},
                xlabel={$x$},
                ylabel={$y$},
                line width=1pt,} }

\usepackage[T1]{fontenc}

\usepackage[sc]{mathpazo}    

\theoremstyle{plain}

\newtheorem{proposition}{Proposition}
\newtheorem{lemma}{Lemma}
\newtheorem{corollary}{Corollary}

\newtheorem{definition}{Definition}

\theoremstyle{definition}

\theoremstyle{remark}
\newtheorem{example}{Example}

\newtheorem*{remark}{Remark}

\newtheorem*{example-1-c}{Example 1 (continued)}
\newtheorem*{example-2-c}{Example 2 (continued)}

\theoremstyle{plain}

\theoremstyle{remark}

\newcommand{\R}{\mathbb{R}}

\newcommand{\Prob}{\mathbb{P}}

\newcommand{\cali}[1]{\mathcal{#1}}
\newcommand{\scr}[1]{\mathscr{#1}}

\newcounter{BoxedDef}

\newcounter{BoxedProp}

\newcounter{BoxedThm}

\newcounter{BoxedLemma}

\newcounter{BoxedCor}

\newcounter{BoxedExample}

\newcommand{\Rev}{\operatorname{Rev}}
\newcommand{\MVR}{\operatorname{MVR}}
\newcommand{\Menu}{\operatorname{Menu}}
\newcommand{\GFOR}{\text{GFOR}}

\newcommand{\cl}[1]{\operatorname{cl}(#1)}



%% file: revenue-gaps.bib
@inproceedings{Psomas:2022rd,
author = {Psomas, Alexandros and Schvartzman, Ariel and Weinberg, S. Matthew},
title = {On infinite separations between simple and optimal mechanisms},
year = {2022}, isbn = {9781713871088}, publisher = {Curran Associates Inc.}, address = {Red Hook, NY, USA}, booktitle = {Proceedings of the 36th International Conference on Neural Information Processing Systems},articleno = {348}, numpages = {12}, location = {New Orleans, LA, USA}, series = {NIPS '22} }

@book{Shaked2007rt,
	address = {New York, {NY}},
	author = {Shaked, Moshe and Shanthikumar, J. George},
	publisher = {Springer},
	title = {Stochastic Orders},
	year = {2007}}

@article{Rubinstein:2018fc,
	address = {New York, NY, USA},
	articleno = {19},
	author = {Rubinstein, Aviad and Weinberg, S. Matthew},
	doi = {10.1145/3105448},
	issn = {2167-8375},
	issue_date = {November 2018},
	journal = {ACM Trans. Econ. Comput.},
	month = oct,
	number = {3--4},
	numpages = {25},
	publisher = {Association for Computing Machinery},
	title = {Simple Mechanisms for a Subadditive Buyer and Applications to Revenue Monotonicity},
	url = {https://doi.org/10.1145/3105448},
	volume = {6},
	year = {2018}}

@inproceedings{Chawla:2007fu,
	address = {New York, NY, USA},
	author = {Chawla, Shuchi and Hartline, Jason D. and Kleinberg, Robert},
	booktitle = {Proceedings of the 8th ACM Conference on Electronic Commerce},
	doi = {10.1145/1250910.1250946},
	isbn = {9781595936530},
	location = {San Diego, California, USA},
	numpages = {9},
	pages = {243--251},
	publisher = {Association for Computing Machinery},
	series = {EC '07},
	title = {Algorithmic pricing via virtual valuations},
	url = {https://doi.org/10.1145/1250910.1250946},
	year = {2007}}

@article{Babaioff:2020kb,
	address = {New York, NY, USA},
	articleno = {24},
	author = {Babaioff, Moshe and Immorlica, Nicole and Lucier, Brendan and Weinberg, S. Matthew},
	doi = {10.1145/3398745},
	issn = {0004-5411},
	issue_date = {August 2020},
	journal = {J. ACM},
	month = jun,
	number = {4},
	numpages = {40},
	publisher = {Association for Computing Machinery},
	title = {A Simple and Approximately Optimal Mechanism for an Additive Buyer},
	url = {https://doi.org/10.1145/3398745},
	volume = {67},
	year = {2020}}

@article{Cheng:2021ew,
	address = {New York, NY, USA},
	articleno = {10},
	author = {Cheng, Yu and Gravin, Nick and Munagala, Kamesh and Wang, Kangning},
	doi = {10.1145/3434419},
	issn = {2167-8375},
	issue_date = {June 2021},
	journal = {ACM Trans. Econ. Comput.},
	month = jan,
	number = {2},
	numpages = {25},
	publisher = {Association for Computing Machinery},
	title = {A Simple Mechanism for a Budget-Constrained Buyer},
	url = {https://doi.org/10.1145/3434419},
	volume = {9},
	year = {2021}}

@unpublished{BenMoshe:2024ej,
	author = {Ben-Moshe, Ran and Hart, Sergiu and Nisan, Noam},
	doi = {10.48550/arxiv.2210.17150},
	eprint = {2210.17150},
	note = {arXiv},
	title = {Monotonic Mechanisms for Selling Multiple Goods},
	year = {2024}}

@article{Boulatov:2021ry,
	author = {Boulatov, Alexei and Severinov, Sergei},
	doi = {10.1016/j.geb.2021.02.001},
	issn = {0899-8256},
	journal = {Games and Economic Behavior},
	pages = {155--178},
	title = {{Optimal and Efficient Mechanisms with Asymmetrically Budget Constrained Buyers}},
	volume = {127},
	year = {2021}}

@article{Babaioff:2022by,
	author = {Babaioff, Moshe and Gonczarowski, Yannai A. and Nisan, Noam},
	doi = {10.1016/j.geb.2021.03.001},
	eprint = {1604.06580},
	issn = {0899-8256},
	journal = {Games and Economic Behavior},
	pages = {281--307},
	title = {{The Menu-Size Complexity of Revenue Approximation}},
	volume = {134},
	year = {2022}}

@article{McAfee:1988rh,
	author = {McAfee, R.Preston and McMillan, John},
	doi = {10.1016/0022-0531(88)90135-4},
	issn = {0022-0531},
	journal = {Journal of Economic Theory},
	number = {2},
	pages = {335--354},
	title = {{Multidimensional Incentive Compatibility and Mechanism Design}},
	volume = {46},
	year = {1988}}

@article{Rochet:1998ju,
	author = {Rochet, Jean-Charles and Chone, Philippe},
	doi = {10.2307/2999574},
	issn = {0012-9682},
	journal = {Econometrica},
	number = {4},
	pages = {783--826},
	title = {{Ironing, Sweeping, and Multidimensional Screening}},
	volume = {66},
	year = {1998}}

@article{Briest:2015ob,
	author = {Briest, Patrick and Chawla, Shuchi and Kleinberg, Robert and Weinberg, S. Matthew},
	doi = {10.1016/j.jet.2014.04.011},
	issn = {0022-0531},
	journal = {Journal of Economic Theory},
	pages = {144--174},
	title = {{Pricing Lotteries}},
	volume = {156},
	year = {2015}}

@inproceedings{Babaioff:2018jy,
	author = {Babaioff, Moshe and Nisan, Noam and Rubinstein, Aviad},
	booktitle = {Proceedings of the 2018 ACM Conference on economics and computation},
	isbn = {1450358292},
	language = {eng},
	pages = {429--429},
	publisher = {ACM},
	series = {EC '18},
	title = {Optimal Deterministic Mechanisms for an Additive Buyer},
	year = {2018}}

@article{Li:2013er,
	author = {Xinye Li and Andrew Chi-Chih Yao},
	issn = {00278424, 10916490},
	journal = {Proceedings of the National Academy of Sciences of the United States of America},
	number = {28},
	pages = {11232--11237},
	publisher = {National Academy of Sciences},
	title = {On Revenue Maximization for Selling Multiple Independently Distributed Items},
	url = {http://www.jstor.org/stable/42712717},
	urldate = {2023-06-06},
	volume = {110},
	year = {2013}}

@article{Hart:2019it,
	author = {Hart, Sergiu and Reny, Philip J.},
	doi = {10.1145/3369927},
	eprint = {1712.08973},
	issn = {2167-8375},
	journal = {ACM Transactions on Economics and Computation},
	number = {4},
	pages = {18},
	title = {{The Better Half of Selling Separately}},
	volume = {7},
	year = {2019}}

@article{Hart:2017bh,
	author = {Hart, Sergiu and Nisan, Noam},
	doi = {10.1016/j.jet.2017.09.001},
	eprint = {1204.1846},
	issn = {0022-0531},
	journal = {Journal of Economic Theory},
	pages = {313--347},
	title = {{Approximate Revenue Maximization with Multiple Items}},
	volume = {172},
	year = {2017}}

@article{Hart:2015ob,
	author = {Hart, Sergiu and Reny, Philip J.},
	doi = {10.3982/te1517},
	issn = {1933-6837},
	journal = {Theoretical Economics},
	number = {3},
	pages = {893--922},
	title = {{Maximal Revenue with Multiple Goods: Nonmonotonicity and Other Observations}},
	volume = {10},
	year = {2015}}

@article{Manzini:2007oy,
	author = {Manzini, Paola and Mariotti, Marco},
	journal = {American Economic Review},
	number = {5},
	pages = {1824--1839},
	title = {{Sequentially Rationalizable Choice}},
	volume = {97},
	year = {2007}}

@article{Kotowski:2020bh,
	author = {Kotowski, Maciej H.},
	doi = {10.3982/te2982},
	issn = {1933-6837},
	journal = {Theoretical Economics},
	number = {1},
	pages = {199--237},
	title = {{First-Price Auctions with Budget Constraints}},
	volume = {15},
	year = {2020}}

@article{Roughgarden:2019vo,
	author = {Roughgarden, Tim and Talgam-Cohen, Inbal},
	doi = {10.1146/annurev-economics-080218-025607},
	issn = {1941-1383},
	journal = {Annual Review of Economics},
	pages = {355---381},
	title = {{Approximately Optimal Mechanism Design}},
	volume = {11},
	year = {2019}}

@article{Daskalakis:2018zg,
	article = {20},
	author = {Constantinos Daskalakis and Nikhil R. Devanur and S. Matthew Weinberg},
	journal = {ACM Transactions on Economics and Computation},
	number = {3-4},
	pages = {1--19},
	title = {Revenue Maximization and Ex-Post Budget Constraints},
	volume = {6},
	year = {2018}}

@incollection{Feng:2022qm,
	author = {Yiding Feng and Jason D. Hartline and Yingkai Li},
	booktitle = {Proceedings of the 2023 Annual ACM-SIAM Symposium on Discrete Algorithms (SODA)},
	institution = {Microsoft Research},
	pages = {3802--3816},
	title = {Simple Mechanisms for Agents with Non-Linear Utilities},
	year = {2023}}

@article{Hart:2019ny,
	author = {Sergiu Hart and Noam Nisan},
	journal = {Journal of Economic Theory},
	pages = {991--1029},
	title = {Selling Multiple Correlated Goods: Revenue Maximization and Menu-Size Complexity},
	volume = {183},
	year = {2019}}

@inproceedings{Chawla:2011mf,
	author = {Chawla, Shuchi and Malec, David and Malekian, Azarakhsh},
	booktitle = {Proceedings of the 12th ACM Conference on Electronic Commerce},
	isbn = {9781450302616},
	language = {eng},
	pages = {253--262},
	publisher = {ACM},
	series = {EC'11},
	title = {Bayesian Mechanism Design for Budget-Constrained Agents},
	year = {2011}}

@article{Richter:2019by,
	author = {Michael Richter},
	journal = {Games and Economic Behavior},
	pages = {30--47},
	title = {Mechanism Design with Budget Constraints and a Population of Agents},
	volume = {115},
	year = {2019}}

@incollection{Cramton:2010txp,
	author = {Peter Cramton},
	booktitle = {The Taxation of Petroleum and Minerals: Principles, Problems and Practice},
	chapter = {10},
	editor = {Philip Daniel and Michael Keen and Charles McPherson},
	publisher = {Routledge},
	title = {How to Best Auction Natural Resources},
	year = {2010}}

@article{Salant:1997vx,
	author = {Salant, David J},
	journal = {Journal of Economics {\&} Management Strategy},
	number = {3},
	pages = {549--572},
	title = {{Up in the Air: GTE{\textquoteright}s Experience in the MTA Auction for Personal Communication Services Licenses}},
	volume = {6},
	year = {1997}}

@inproceedings{Devanur:2017ik,
	address = {New York, New York, USA},
	author = {Devanur, Nikhil R and Weinberg, S Matthew},
	booktitle = {EC 2017 - Proceedings of the 2017 ACM Conference on Economics and Computation},
	month = jun,
	organization = {Microsoft Research},
	pages = {39--40},
	publisher = {ACM Press},
	title = {{The Optimal Mechanism for Selling to a Budget-Constrained Buyer: the General Case}},
	year = {2017}}

@article{Myerson1981,
	author = {Myerson, Roger B},
	journal = {Mathematics of Operations Research},
	number = {1},
	pages = {58--73},
	title = {{Optimal Auction Design}},
	volume = {6},
	year = {1981}}

@inproceedings{Bhattacharya:2010vy,
	author = {Bhattacharya, Sayan and Conitzer, Vincent and Munagala, Kamesh and Xia, Lirong},
	booktitle = {Proceedings of the Twenty-First Annual ACM-SIAM Symposium on Discrete Algorithms},
	pages = {554--572},
	title = {{Incentive Compatible Budget Elicitation in Multi-Unit Auctions}},
	year = {2010}}

@article{Laffont:1996ig,
	author = {Laffont, J J and Robert, J},
	journal = {Economics Letters},
	number = {2},
	pages = {181--186},
	title = {{Optimal Auction with Financially Constrained Buyers}},
	volume = {52},
	year = {1996}}

@article{Che:2013cu,
	author = {Che, Yeon-Koo and Gale, Ian and Kim, Jinwoo},
	journal = {Review of Economic Studies},
	month = feb,
	number = {1},
	pages = {73--107},
	title = {{Assigning Resources to Budget-Constrained Agents}},
	volume = {80},
	year = {2013}}

@article{Che:1998je,
	author = {Che, Yeon-Koo and Gale, Ian},
	journal = {Review of Economic Studies},
	month = jan,
	number = {1},
	pages = {1--21},
	title = {{Standard Auctions with Financially Constrained Bidders}},
	volume = {65},
	year = {1998}}

@article{Pai:2014cm,
	author = {Pai, Mallesh M and Vohra, Rakesh V},
	journal = {Journal of Economic Theory},
	month = mar,
	pages = {383--425},
	title = {{Optimal Auctions with Financially Constrained Buyers}},
	volume = {150},
	year = {2014}}

@article{Che:2000ux,
	author = {Che, Yeon-Koo and Gale, Ian},
	journal = {Journal of Economic Theory},
	pages = {198--233},
	title = {{The Optimal Mechanism for Selling to a Budget-Constrained Buyer}},
	volume = {92},
	year = {2000}}
